\documentclass{llncs}

\usepackage{amssymb}
\usepackage{graphicx}
\usepackage{multirow}
\usepackage{multicol}
\usepackage{theorem,latexsym}
\usepackage{amsmath,amssymb,enumerate}
\usepackage{xspace}
\usepackage{chngpage} 

\graphicspath{{c:/latex_figures/}}

\newenvironment{proofof}[1]{\smallskip\noindent{\bf Proof of #1}}%
        {\hspace*{\fill}$\Box$\par}

\newtheorem{observation}[theorem]{Observation}
\newtheorem{invariant}[theorem]{Invariant}

\renewenvironment{proof}{\vspace{-0.05in}\noindent{\bf Proof:}}
        {\hspace*{\fill}$\Box$\par}

        {\hspace*{\fill}$\Box$\par}

\newcommand{\edcs}{{\mbox edge degree constrained subgraph}}
\newcommand{\edcsab}{{\rm EDCS}}

\renewcommand{\dh}{d_H}
\newcommand{\al}{\alpha}
\newcommand{\be}{\beta}
\newcommand{\bem}{\beta^-}

\newcommand{\eps}{\epsilon}
\newcommand{\oeps}{(1 + \eps)}
\newcommand{\DE}{\Delta}
\newcommand{\De}{\Delta}
\newcommand{\la}{\lambda}

\newcommand{\ur}{{\rm ur}}
\newcommand{\ignore}[1]{}

\newcommand{\ed}{{\delta}}  
\newcommand{\owned}{O}
\newcommand{\full}{F}
\newcommand{\defic}{E}
\newcommand{\td}{\tilde{d}}
\newcommand{\tdu}{\td_u}

\newcommand{\bu}{B^u}

\newcommand{\sm}{\setminus}

\newcommand{\etal}{{\em et al.}}

\newcommand{\comment}[1]{}

\newcommand{\negsp}{\vspace*{-0.15in}}

\begin{document}


\title{Fully Dynamic Matching in Bipartite Graphs}

\author{
 Aaron Bernstein\thanks{Supported in part by an NSF Graduate Fellowship and a Simons Foundation Graduate Fellowship. \tt{bernstei@gmail.com}}  \and
 Cliff Stein\thanks{Supported in part by NSF grants CCF-1349602 and CCF-1421161. \tt{cliff@ieor.columbia.edu}}}

\institute{Department of Computer Science, Department of IEOR,  Columbia University}

\authorrunning{Bernstein and Stein}
\maketitle

\begin{abstract}
\negsp
We present two fully dynamic algorithms for maximum cardinality matching in bipartite graphs. 
Our main result is a {\em deterministic} algorithm that maintains a $(3/2 +
\eps)$ approximation in {\em worst-case} update time $O(m^{1/4}\eps^{-2.5})$, 
which is polynomially faster than all previous \emph{deterministic} algorithms for \emph{any} constant approximation,
and faster than all previous algorithms (randomized included) that achieve a better-than-2 approximation.
We also give stronger results for bipartite graphs whose arboricity is at most $\al$, achieving a $(1+
\eps)$ approximation in worst-case update time $O(\al (\al + \log n))$ for constant
$\eps$. Previous results for small arboricity graphs 
had similar update times but could only maintain a maximal matching (2-approximation).
All these previous algorithms, however, were not limited to bipartite graphs.
\negsp
\end{abstract}

\negsp
\section{Introduction} 
\label{sec:introduction}

The problem of finding a maximum cardinality matching in a bipartite
graph is a classic problem in computer science and combinatorial
optimization. There are efficient polynomial time algorithms
(e.g. \cite{HopcroftK73}), and well-known applications, ranging from
early algorithms to minimize transportation costs
(e.g. \cite{Hitchcock41,Kantorovitch42}) and including recent
applications in the area of on-line advertising and social media
(e.g. \cite{MehtaSVV05,FHKMS10}).  We observe that for matching, the
restriction to bipartite graphs is natural and still models many
real-world applications and also that in many of these applications,
the graph is actually changing over time.  We study the \emph{fully
  dynamic} variant of the maximum cardinality matching problem in
which the goal is to maintain a near-maximum matching in a graph
subject to a sequence of edge insertions and deletions.  When an edge
change occurs, the goal is to maintain the matching in time
significantly faster than simply recomputing it from scratch.

One of our results is for bipartite {\em small-arboricity} graphs,
which we define here. The {\em arboricity} of an $n$-node $m$-edge  graph, denoted by
$\alpha(G)$ is $\max_J \frac{|E(J)|}{V(J)-1}$ where $J = (V (J),
E(J))$ is any subgraph of $G$ induced by at least two vertices. Many
classes of graphs in practice have constant arboricity, including
planar graphs, graphs with bounded genus and graphs with bounded tree
width.  Every graph has arboricity at most $O(\sqrt{m})$.

\subsection{Previous Work}
\label{sec:previous}
In addition to exact algorithms on static graphs, there is previous
work on approximating matching and on finding online matchings. 
Duan and Pettie showed how to find a $(1+\eps)$-approximate
weighted matching in nearly linear time \cite{DuanP14}; their paper
also contains an excellent summary of the history of matching
algorithms.  Motivated partly by online advertising, there has also
been significant work on ``online matching''
(e.g. \cite{MehtaSVV05,FHKMS10}), both exact and approximate.  In most
online matching work, the graph is dynamic, but with a restricted set
of updates.  Typically, one side of the bipartite graph is fixed at
the beginning of the algorithm.  The vertices on the other side
arrive, one at a time, and when a vertex arrives, we learn about all
of its incident edges.  Deletions are not allowed, nor typically are
changes to the matching, although some work also studies models that measure the number of changes needed
to maintain a matching \cite{CDKL09,GuptaKS14,BosekLSZ14}.

We now turn to fully dynamic matchings.
Algorithms can be classified by update time, approximation ratio, whether
they are randomized or deterministic and whether they have a worst-case or amortized
update time. The distinction between deterministic and randomized is particularly
important here as all of the existing
randomized algorithms require the assumption of an {\em oblivious}
 adversary that does not see the algorithm's random bits; thus, in
addition to working only with high probability, randomized dynamic
algorithms must make an extra assumption on the model which makes them
 inadequate in certain settings. 

For maintaining an {\em exact} maximum matching, the best known update
time is $O(n^{1.495})$ (Sankowski \cite{Sankowski07}), which in dense
graphs is much faster than reconstructing the matching from
scratch. If we restrict the model to bipartite graphs and to the
incremental or decremental setting -- where we allow only edge
insertions or only edge deletions (but not both) -- Bosek \etal
\cite{BosekLSZ14} show that we can achieve {\em} total update time
(over all insertions or all deletions) $m\sqrt{n}$ for an exact
matching and $m\eps^{-1}$ for a $\oeps$-matching, which is optimal in
the sense that it matches the best known bounds for the static case. For the
special case of \emph{convex} bipartite graphs in the fully dynamic setting, 
Brodal \etal~\cite{BrodalGHK07}, showed how to maintain an \emph{implicit} (exact) matching 
with very fast update but slow query time.  

Returning to the general problem of maintaining an explicit matching
in a fully dynamic setting, we can achieve a much faster update time than
$O(n^{1.495})$ if we allow approximation. One can trivially
maintain a {\em maximal} (and so 2-approximate) matching in $O(n)$ time per update.
Ivkovic and Lloyd \cite{IvkovicL93} showed how to improve the update time to
$O((m + n)^{\sqrt{2}/2})$. Onak and Rubinfeld \cite{OnakR10} were to
first to achieve truly fast update times, presenting a
randomized algorithm that maintains a $O(1)$-approximate matching in
amortized update time $O(\log^2 n)$ time (with high
probability). Baswana \etal \cite{BaswanaGS11} improved upon this with
a randomized algorithm that maintains a maximal matching
(2-approximation) in amortized update $O(\log n)$ time per
update. These two algorithms are extremely fast, but suffer
from being amortized and inherently randomized, and also from the fact
their techniques focus on local changes, and so seem unable to break
through the barrier of a 2-approximation.

The first result to achieve a better-than-2 approximation was by
Neiman and Solomon \cite{NeimanS13}, who presented a {\em
  deterministic, worst-case} algorithm for maintaining a
3/2-approximate matching. However, the price of this improvement was a
huge increase in update time: from $O(\log n)$ to $O(\sqrt{m})$. Gupta and
Peng \cite{GuptaP13} later improved upon the approximation, presenting
a deterministic algorithm that maintains a $\oeps$-approximate
matching in worst-case update time $O(\sqrt{m}\eps^{-2})$ (the same paper
achieves an analogous result for maintaining a near-maximum {\em
  weighted} matching in weighted graphs).

The two deterministic algorithms are strongly tethered to
the $\sqrt{m}$ bound and do not seem to contain any techniques for breaking past
it. An important open question was thus: can we achieve $o(\sqrt{m})$
with a deterministic algorithm? (In fact Onak and
Rubinfeld \cite{OnakR10} presented a deterministic algorithm with
amortized update time $O(\log^2 n)$, but it only achieves a
$\log(n)$-approximation.) Very recently, Bhattacharya, Henzinger, and
Italiano \cite{BhattacharyaHI15} presented a deterministic algorithm
with worst-case update time $O(m^{1/3}\eps^{-2})$ that maintains a $(4
+ \eps)$ approximation; this can be improved to $(3 + \eps)$ at the
cost of introducing amortization. The same paper presents a
deterministic algorithm with amortized update time only
$O(\eps^{-2}\log n)$ that maintains a $(2 + \eps)$ {\em fractional}
matching. Finally, Neiman and Solomon \cite{NeimanS13} showed that in
graphs of constant arboricity we can maintain a maximal (so
2-approximate) matching in amortized time $O(\log(n)/\log\log(n))$;
using a recent dynamic orientation algorithm of Kopelowitz \etal
\cite{KopelowitzKPS14}, this algorithm yields a $O(\log(n))$ {\em
  worst-case} update time.

Very recently there have been some conditional lower bounds
for dynamic approximate matching.
Kopelowitz \etal \cite{KopelowitzPP14} show that assuming $3$-sum hardness
any algorithm that maintains a matching in which all 
augmenting paths have length at least $6$ requires
an update time of $\Omega(m^{1/3}-\zeta)$ for any fixed $\zeta > 0$.
Henzinger \etal show that such an algorithm in fact requires $\Omega(m^{1/2}-\zeta)$ time
if one assumes the Online Matrix-Vector conjecture.

\negsp
\subsection{Results}
\label{sec:results}

If we disregard special cases such as
small arboricity or fractional matchings, we see that existing
algorithms for dynamic matching seem to fall into two groups: there
are fast (mostly randomized) algorithms that do not break through the
2-approximation barrier, and there are slow algorithms with $O(\sqrt{m})$
update that achieve a better-than-2 approximation. Thus the obvious
question is whether we can design an algorithm -- deterministic or
randomized -- that achieves a tradeoff between these two: a
$o(\sqrt{m})$ update time and a better-than-2 approximation. We
answer this question in the affirmative for bipartite graphs.

\begin{theorem}
\label{thm:main-general}
Let $G$ be a  bipartite graph subject to a series of edge insertions and deletions, and let $\eps < 2/3$. 
Then, we can maintain a $(3/2+\eps)$-approximate matching in $G$
in deterministic worst-case update time $O(m^{1/4}\eps^{-2.5})$. 
\end{theorem}

This theorem achieves a new trade-off even if one considers existing
randomized algorithms. Focusing on only deterministic algorithms the
improvement is even more drastic: our algorithm improves upon not just
$\sqrt{m}$ but $m^{1/3}$, and so achieves the fastest known
deterministic update time (excluding the $\log(n)$-approximation of
\cite{OnakR10}), while still maintaining a better-than-2
approximation. Also, since $m^{1/4} = O(\sqrt{n})$, our algorithm is
the first to achieve a better-than-2 approximation in time strictly
sublinear in the number of nodes. Of course, our algorithm has the
disadvantage of only working on bipartite graphs.

For small arboricity graphs we also show how to break through the maximal matching (2-approximation) barrier and achieve a $\oeps$-approximation.

\begin{theorem}
\label{thm:main-small}
Let $G$ be a bipartite graph subject to a series of edge insertions
and deletions, and let $\eps < 1$. Say that at all
times $G$ has arboricity at most $\alpha$. Then, we can maintain a
$(1+\eps)$-approximate matching in $G$ in deterministic worst-case update time
$O(\al(\al + \log(n)) + \eps^{-4}(\al + \log(n)) + \eps^{-6})$  
For constant $\alpha$ and $\eps$ the update time is $O(\log(n))$, and for
$\al$ and $\eps$ polylogarithmic the update time is polylogarithmic. 
\end{theorem}

\noindent Note that a $\oeps$-approximation with polylog update time is pretty
much the best we can hope for. The conditional lower bound of Abboud
and Williams \cite{AbboudW14} provides a strong indication that such a result
is likely not possible for general graphs, but we have presented the
first class of graphs (bipartite, polylog arboricity) for which it is
achievable.

\paragraph{Remark:}
This paper is the full version of an extended abstract that appeared 
in ICALP 2015 \cite{BernsteinS15}. The conference version, however, has a mistake:
all the main results are correct as stated, but Lemma 2 in Section 4 of that paper is false. Lemma 2 was not a significant result in and of itself, but was used as a building block for later theorems proved in the section.
In this paper we present a modified version of Section 4. Most of the building blocks and the overall structure are the same as before, except for Lemma \ref{lem:path-degree}; this lemma is new to the current version, and allows us to avoid relying on the very particular partition (falsely) indicated by Lemma 2 of the previous version.

\negsp
\subsection{Techniques} \label{sec:techniques}
We can think of the dynamic matching problem as follows: We are given
a dynamic graph $G$ and want to maintain a large subgraph $M$ of
maximum degree 1. This task turns out to be quite hard because, as the graph
evolves,  $M$ is
unstable and has few appropriate structural properties. 

Very recently, Bhattacharya \etal \cite{BhattacharyaHI15} presented the idea of using a 
transition subgraph $H$, which they refer to as a \emph{kernel} of $G$: the idea is to maintain 
$H$ as $G$ changes, and then maintain $M$ in $H$. Maintaining an approximate matching $M$ is 
significantly easier in a bounded degree graph, so we need a graph $H$ that has the following
properties: it should have bounded degree, it should be easy to
maintain in $G$, and most importantly, a large matching using
edges in $H$ should be a good approximation to the maximum matching in $G$.

Our algorithm uses the same basic idea of transition subgraph with 
bounded degree, but the details are entirely different from those in \cite{BhattacharyaHI15} . 
Their subgraph $H$ is just a maximal $B$-matching with $B$ around $m^{1/3}$, that allows some
slack on the maximality constraint. The use of a maximal matching is a natural
choice in a dynamic setting because maximality is a purely \emph{local} constraint, and so easier to maintain dynamically. 
The downside is that as long as one relies on maximality,
one can never achieve a better-than-2 approximation; due to other difficulties, their paper 
in fact only achieves a $(3 + \eps)$-approximation.

The main technical contribution of this paper is to present a new type of bounded-degree subgraph, 
which we call an {\em \edcs} (\edcsab). The problem with a simple B-matching is that the
edges are not sufficiently ``spread out'' to all the vertices: imagine
that $G$ consists of 4 sets $L_1, L_2, R_1, R_2$, each of size $n/2$,
where the edges form a complete graph except that there are no edges
between $L_2$ and $R_2$. One possible maximal B-matching includes many
edges between $L_1$ and $R_1$ while leaving $L_2$ and $R_2$ completely
isolated.  The resulting matching is only 2-approximate, which is what we are trying to overcome.
Our \edcsab\ circumvents this problem by trying  to spread out edges.
For each edge, instead of separately upper bounding the matching-degree of each endpoint (B-matching) 
it upper bounds the {\em sum} of
the matching-degrees of the endpoints, and then captures the notion of maximality
by also lower bounding this sum for edges not in the matching.
Using an \edcsab\ prevents the above scenario as the sum of the matching-degrees of edges from $L_1$ to $R_2$ 
will be illegally small unless the matching-degree of $R_2$ is raised by adding some of those edges to the graph,
thus ensuring a larger matching in $H$.

Although the definition is somewhat similar, the structure of an
\edcs\ is entirely different from that of a maximal B-matching, and
for this reason both our analysis of the approximation factor and our
algorithm for maintaining this subgraph are entirely different from
those in \cite{BhattacharyaHI15}.  In particular, while the
constraints in an \edcsab\ seem purely local in that they concern only
the degrees of the endpoints of an edge, they in fact have a global
effect in a way that they do not in a maximal B-matching. In the
latter, as long as an edge does not directly violate the degree
constraints, it can \emph{always} be added to the maximal B-matching,
without concern for the edges elsewhere in the graph. But as seen from
the above example, this is not true in an \edcsab: although the edges
from $L_1$ and $R_1$ do not themselves violate any constraints, they
prevent the constraints between $L_1$ and $R_2$ or $L_2$ and $R_1$
from being satisfied.  An analysis of this global structure is what
allows us to go beyond the 2-approximation.  On the other hand, the
same global structure makes the \edcsab\ more difficult to maintain
dynamically; we end up showing that an \edcsab\ contains something
akin to augmenting paths, although more locally well behaved.  We also
develop a general new technique for maintaining a transition subgraph
based on dynamic graph orientation, which allows us to reduce the
update time from $O(m^{1/3})$ to $O(m^{1/4})$.  That being said, the
additional complications inherent in an \edcsab\ have so far prevented
us from extending our results to non-bipartite graphs.

\negsp
\section{Preliminaries}
\label{sec:preliminaries}
\negsp

Let $G = (L \bigcup R, E)$ be an undirected, unweighted bipartite
graph where $|L| = |R| = n$ and $|E| = m$. Unless otherwise specified,
``graph" will always refer to a bipartite graph. In general, we will
often be dealing with graphs other than $G$, so all of our notation
will be explicit about the graph in question. We define $d_G(v)$ to be
the degree of a vertex $v$ in $G$; if the graph in question is
weighted, then $d_G(v)$ is the sum of the weights of all incident
edges. We define {\em edge degree} as $\ed(u,v) = d(u) + d(v)$. If $H$ is a subgraph of $G$, we say that an edge in $G$ is
\emph{used} if it is also in $H$, and \emph{unused} if it is not in
$H$. Throughout this paper we will only be dealing with subgraphs $H$
that contain the full vertex set of $G$, so we will use the notion of
a subgraph and of a subset of edges of $G$ interchangeably.

A matching in a graph $G$ is a set of disjoint edges in $G$. We let
$\mu(G)$ denote the size of the maximum matching in $G$. A vertex is
called {\em matched} if it is incident to one of the sets in the
matching, and {\em free} or {\em unmatched} otherwise. 
We now state a simple corollary of an existing result of \cite{GuptaP13}.

\begin{lemma}[\cite{GuptaP13}]
\label{lem:maxdegree}
If a dynamic graph $G$ has maximum degree B at all times, then we can
maintain a $\oeps$-approximate matching under insertions and deletions in
worst-case update time $O(B\eps^{-2})$ per update. 
\end{lemma}
\negsp
\begin{proof} This lemma immediately follows from a simple
algorithm presented in Section 3.2 of \cite{GuptaP13} which shows how to achieve
update time $|E(G)|\eps^{-2}/\mu(G)$ (for the transition from worst-case to
amortized see appendix A.3 of the same paper), as well as the fact
that we always have $|E(G)|/\mu(G) \leq 2B$ because all edges must be
incident to one of the $2\mu(G)$ matched vertices in the maximum
matching, and each of those vertices have degree at most $B$.
\end{proof}

\negsp
\paragraph{Orientations}
An orientation of an undirected graph $G$ is an assignment of a
direction to each edge in $E$. 
Given an orientation of edge $(u,v)$ from $u$ to $v$, 
we will say that $u$ {\em owns}
edge $(u,v)$ and will define the {\em load} of a vertex $v$ to be the number
of edges owned by $v$. Orientations of small max load are closely
linked to arboricity: every graph with arboricity $\al$ has an
$\al$-orientation \cite{NashW61}. Our algorithms will at all times maintain an orientation 
of the {\em dynamic} graph $G$. The details are in Section \ref{sec:ap-sqrt-orientation}, but for the sake of intuition,
it suffices to say the following: for a graph with small arboricity $\al$, an existing result of
Kopelowitz \etal \cite{KopelowitzKPS14} dynamically maintains an orientation with small max-load 
and small worst-case update time (Theorem \ref{thm:dynamic-orientation-arb}); for arbitrary graphs, we present a new result that maintains
a max load of $O(\sqrt{m})$ in $O(1)$ worst-case update time (Theorem \ref{thm:dynamic-orientation-general}). 

\negsp
\section{The Framework}
\label{sec:framework}
\negsp

We now define the transition subgraph $H$ mentioned in Section \ref{sec:techniques}.

\begin{definition}
\label{dfn:unweighted-edcs}
An unweighted \edcs (\edcsab) $(G,\beta,\bem)$ is a subset of the edges $H \subseteq E$ with the following 
properties:
\\(P1) if $(u,v)$ is used (in $H$) then $\dh(u) + \dh(v) \leq \be$ ,
\\(P2)  if $(u,v)$ is unused (in $G - H$) then $\dh(u) + \dh(v) \geq \bem$.
\end{definition}

\noindent We also define a similar subgraph where edges in $H$ have
weights, effectively allowing them to be used more than once. The
properties change somewhat as now used edges can always take more
weight, so it makes sense to lower bound the degrees of used edges as
well.  Recall that the degree of a vertex in a weighted graph is the sum of the weights
of the incident edges.
 
\begin{definition}
\label{dfn:weighted-edcs}
A weighted \edcs (\edcsab) $(G,\beta,\bem)$ is a subset of the edges $H \subseteq E$ with positive integer weights that has properties:
\\(P1)  if $(u,v)$ is used then $\dh(u) + \dh(v) \leq \be$
\\(P2)  for {\em all} edges $(u,v)$, we have  $\dh(u) + \dh(v) \geq \bem$
\end{definition}

\paragraph{Algorithm Outline:} To process an edge insertion/deletion in $G$: 
firstly, we update the small-max-load edge orientation (Theorem
\ref{thm:dynamic-orientation-arb} or  
\ref{thm:dynamic-orientation-general} in Appendix \ref{sec:ap-sqrt-orientation}).

Secondly, we update the subgraph $H$ so it remains a valid \edcsab\ of the changed graph $G$ (Section~\ref{sec:edcs});
this relies on the graph orientation for efficiency. 

Thirdly, we update the $\oeps$-approximate matching in $H$ with respect to the
changes to $H$ from the previous step (See Lemma \ref{lem:maxdegree}). 
The maintained $\oeps$-approximate matching of $H$ is also our final matching in $G$; 
the central claim of this paper is that because $H$ is an \edcsab, 
$\mu(H)$ is not too far from $\mu(G)$, 
so a good approximation to $\mu(H)$ is also a decent approximation to $\mu(G)$ (see Section \ref{sec:matching}). 

\ignore{
\noindent Below is an outline of how we process an edge insertion/deletion in $G$:
\begin{enumerate}
\item Update the small-max-load edge orientation (Theorem
  \ref{thm:dynamic-orientation-arb} or  
	\ref{thm:dynamic-orientation-general} in Appendix \ref{sec:ap-orientation} \label{step:orientation}. 
\item Update the subgraph $H$ so it remains a valid \edcs\ of the
  changed graph $G$. This relies on the orientation from step
  \ref{step:orientation} for efficiency. \label{step:edcs}  (See Section~\ref{sec:edcs}.)
\item Update the $\oeps$-approxiamte matching in $H$ with respect to the
  changes in $H$ from step
  \ref{step:edcs}. \label{step:maintain-matching} (See Lemma \ref{lem:maxdegree}.)
\end{enumerate}
}

There is a subtle difficulty that arises from using a transition graph
in a dynamic algorithm. We know from Lemma \ref{lem:maxdegree} that if
the maximum degree in $H$ is guaranteed to always be below $\DE_H$,
then the time to update a $\oeps$-approximate matching in $H$ will be
$O(\DE_H)$ {\em per update in $H$}. But a single change in
$G$ could in theory causes many changes in $H$, each of which would
take $O(\DE_H)$ time to process. This motivates the following
definition:

\begin{definition}
\label{dfn:update-ratio}
Let $H$ be a subgraph of a dynamic graph $G$, and let $A$ be an algorithm that modifies the edges of $H$ as $G$ changes. Then, we define the {\em update ratio} (\ur) of $A$ to be the maximum number of edge changes (insertions or deletions) that could be made to $H$ given a single edge change in $G$.
\end{definition}

\noindent We can now state the main theorems of the paper. We present
general and small arboricity graphs separately, but 
the basic framework described above remains the same
in both cases. In all the theorems below, the parameter $\eps$
corresponds to the desired approximation ratio (either $\oeps$ or $(3/2 + \eps)$).

\negsp
\subsection{General Bipartite Graphs}
\label{sec:framework-general}
For the sake of intuition, think of $\be$ in the two theorems below as roughly $m^{1/4}$.  

\begin{theorem}
\label{thm:match-general}
Let $G$ be a bipartite graph, and let $\la = \eps / 4$.  Let $H$ be an
unweighted \edcsab\ with $\bem = \be(1- \la)$, where $\be$ is a parameter
we will choose later. 
Then $\mu(H) \geq (2/3 -\eps) \mu(G) $.
\end{theorem}

\begin{theorem}
\label{thm:h-general}
Let $G$ be a bipartite graph.  Let $H$ be an unweighted \edcsab\ with
$\bem = \be(1-\la)$, where $\la$ is a positive constant less than
$1$. There is an algorithm that maintains $H$ over updates in $G$
(i.e. maintains $H$ as a valid \edcs) with the following properties:
\begin{itemize}
\item The algorithm has worst-case update time $O( \left( \frac{1}{\lambda} \right) \left(\beta +
\frac{\sqrt{m}}{\la \be} \right) )$.
\item The update ratio of the algorithm is $O(1/\la)$ (see Definition \ref{dfn:update-ratio}).
\end{itemize}
\end{theorem}

\begin{proofof}{Theorem \ref{thm:main-general}} 
We use the algorithm outline presented near 
the beginning of Section \ref{sec:framework}.
We let be transition subgraph $H$ be an unweighted \edcsab($G, \beta, \beta(1 - \la)$) 
with $\la = 4\eps^{-1} = O(\eps^{-1})$ and $\beta = m^{1/4}\eps^{1/2}$. 
By Theorem \ref{thm:h-general} we can maintain $H$ in worst-case update time 
$O( \left( \frac{1}{\lambda} \right) \left(\beta +
\frac{\sqrt{m}}{\la \be} \right) ) = O(m^{1/4}\eps^{-2.5} + m^{1/4}\eps^{-.5}) =
O(m^{1/4}\eps^{-2.5})$. The update ratio is $O(\la^{-1}) = O(\eps^{-1})$. 
Since degrees in $H$ are clearly bounded by $\beta$, by Lemma \ref{lem:maxdegree} 
we can maintain a $\oeps$-approximate matching in $H$ in time
$O(\beta\eps^{-2})$; multiplying by the update ratio of maintaining $H$ in $G$, 
we need $O(\beta\eps^{-3}) = O(m^{1/4}\eps^{-2.5})$ time to 
maintain the matching per change in $G$.
By Theorem \ref{thm:match-small}, $\mu(H)$ is a $(3/2 + \eps)$-approximation to $\mu(G)$, so 
our matching is a $(3/2 + \eps) \oeps = (3/2 + \eps)$-approximate matching in $G$.
\end{proofof}

\negsp
\subsection{Small Arboricity Graphs}

\begin{theorem}
\label{thm:match-small}
Let $G$ be a bipartite graph, and let $\be > 8\eps^{-2}$. Let $H$ be a weighted \edcsab\ with $\bem = \be-1$.  
Then $\mu(H) \geq \mu(G) (1 - \eps)$.
\end{theorem}

\begin{theorem}
\label{thm:h-small}
Let $G$ be a bipartite graph with arboricity $\al$.  Let $H$  be a weighted \edcsab\ with $\bem = \be-1$.  
There is an algorithm that maintains $H$ over updates in $G$ with the following properties:
\begin{itemize}
\item The algorithm has worse-case update time $O(\beta^2(\alpha + \log n) + \al(\alpha + \log n))$ .
\item The update ratio of the algorithm is $O(\be)$ (see Definition \ref{dfn:update-ratio}).

\end{itemize}
\end{theorem}

\noindent The proof of Theorem \ref{thm:main-small} is analogous to that of Theorem \ref{thm:main-general}
with $\beta$ set to $\eps^{-2}$: see Appendix \ref{sec:proof-main-small}

\ignore{
\negsp 
\paragraph{Overview of the rest of the paper} The rest of the
paper will be devoted to proving the four theorems above. Theorems
\ref{thm:match-general} and \ref{thm:match-small} are clearly
analogous and end up requiring many similar proof techniques. In fact,
\ref{thm:match-small} also applies general bipartite
graphs (not just small arboricity ones); we do not apply it 
only because a weighted \edcs\ is too difficult to maintain in
the general setting. Theorems \ref{thm:h-general} and
\ref{thm:h-small} are also analogous. For this reason, our sections
are broken up not by result (general vs. small-arboricity) but by
topic: Section \ref{sec:matching} proves that an \edcsab\ contains a
good matching (Theorems \ref{thm:match-general},
\ref{thm:match-small}), while Section \ref{sec:edcs} shows how to
maintain an \edcs\ efficiently (Theorems \ref{thm:h-general},
\ref{thm:h-small}). 
}
	
\negsp
\section{An \edcsab\ Contains an Approximate Matching}
\label{sec:matching}

\renewcommand{\sl}{S_L}
\newcommand{\sr}{S_R}
\newcommand{\tl}{T_L}
\newcommand{\tr}{T_R}
\newcommand{\sls}{S_L^*}
\newcommand{\trs}{T_R^*}
\newcommand{\tstar}{T^*}
\newcommand{\grm}{G^r_M}
\newcommand{\grmh}{G^r_{M(H)}}

\newcommand{\pl}{P_L}
\newcommand{\pr}{P_R}
\newcommand{\pls}{P_L^*}
\newcommand{\ql}{Q_L}
\newcommand{\qr}{Q_R}
\newcommand{\qrs}{Q_R^*}

In this section we prove Theorems \ref{thm:match-general} and
\ref{thm:match-small}. 
Both proofs will take the form of a proof by contradiction. For
Theorem \ref{thm:match-general} to be false, there must be an
unweighted \edcsab$(G, \beta, \beta(1 - \la))$\ $H$ such that $\mu(H)
< (2/3 - \eps)\mu(G)$; similarly, for Theorem \ref{thm:match-general}
to be false there must a weighted \edcsab$(G, \beta, \beta-1)$\ $H$ such
that $\mu(H) < (1 - \eps)\mu(G)$. 

To exhibit the contradiction for
Theorem \ref{thm:match-general}, we start by
establishing a simple property that must hold of any unweighted \edcsab\ $H$ of $G$
for which $\mu(H)$ is smaller than $\mu(G)$; 
the smaller $\mu(H)$, the more constraining the property.
We then prove a separate lemma which shows that
for small enough $\mu(H)$, this property is impossible to satisfy.
The proof of Theorem \ref{thm:match-small}
follows a similar approach but takes advantage of the fact that
$H$ is now a \emph{weighted} \edcsab\ to prove stronger versions of these claims.
 
We use the convention that the subscript $L$ or $R$
refer to the side of the bipartition in which the vertices lie.
We start by formally defining the standard cut induced
by a maximum matching in a bipartite graph.

\begin{definition}
\label{dfn:matching-cut}
Let $G$ be a bipartite graph and let $M$ be a maximum matching in $G$. 
Let $\grm$ be the residual graph defined with respect to $M$.
We define the cut (P,Q) induced by $M$ to be a
partition of the vertices of $G$ into the following sets:
\begin{itemize}
\item $P = \pls \bigcup \pl \bigcup \pr$ where $\pls$ contains all free vertices in $L$, and $\pl$ and $\pr$ are the matched vertices in $L$ and $R$ that are reachable from $\pls$ in $\grm$. 
\item $Q = \qrs \bigcup \qr \bigcup \ql$ where $\qrs$ contains all free vertices in $R$, and $\ql$ and $\qr$ contain
all matched vertices in $L$ and $R$ that are NOT reachable from $\pls$ in $\grm$. 
\end{itemize}
\end{definition}

\begin{observation}
\label{obs:matching-cut}
Note that in the above definition of the $(P,Q)$ cut (Definition \ref{dfn:matching-cut}),
it is crucial that $M$ is a \emph{maximum} matching:
if $M$ was not maximum then some of the free vertices in $R$ 
might also be reachable from $\pls$ in $\grm$,
which would result in $\pr$ and $\qrs$ not being disjoint.
When $M$ is maximum, however, this issue does not arise
because there cannot be a path in $\grm$ between two free vertices.

The following stems directly from the definition of a $(P,Q)$ cut:
there can be no edge from $P$ to $Q$ in $\grm$.
There can however be backwards edges from $Q$ to $P$ in $\grm$.
Put otherwise, $G$ cannot contain edges between
$\pls \bigcup \pl$ and $\qrs \bigcup \qr$.
\end{observation}

\subsection{General Bipartite Graphs}

In this section we will prove Theorem \ref{thm:match-general}.

\begin{lemma}
\label{lem:cut-sets-unweighted}
Let $G = (V,E_G)$ be a bipartite graph, and let $H = (V,E_H)$ be a
an \emph{unweighted} \edcsab\ with $\bem = \be(1-\la)$ for some $0 < \la < 1$. 
Then there exist disjoint sets of vertices $S,T$ such that
\begin{enumerate}
\item $|T| = \mu(H)$.
\item $|S| = 2(\mu(G) - \mu(H))$.
\item All edges in $H$ incident to $S$ go $T$
\item The average degree $\dh(s)$ of vertices $s \in S$ is at least $\frac{\be(1-\la)}{2}$.
\end{enumerate}
\end{lemma}

\begin{proof}
Let $M(H)$ be some maximum matching in $H$.
$M(H)$ induces a standard $(P,Q)$ cut in the graph $H$ as defined in 
Definition \ref{dfn:matching-cut}, 
which partitions the vertices into sets
Let $\pls, \pl, \pr, \qrs, \ql, \qr$.
We start by setting $T$ to be $\pl \bigcup \qr$. 
This clearly satisfies property 1 because $|\pl| + |\qr| = |\pl| + |\ql|$,
which is equal to $\mu(H)$ because $\pl \bigcup \ql$ is exactly the set of 
vertices in $L$ that are matched in $M(H)$.

Now, let us look at the matching $M(H)$ from the perspective of the larger graph $G$.
Like all matchings, $M(H)$ induces a residual graph $\grmh$ in $G$, 
but $M(H)$ is \emph{not} maximum in $G$ (only in $H$).
Thus, we know that
$G$ contains $\mu(G) - \mu(H)$ vertex-disjoint augmenting
paths from $\pls$ to $\qrs$ in $\grmh$, 
and that each of these paths contains at least one edge crossing the $P-Q$ cut,
i.e. at least one edge from 
$\pls \bigcup \pl$ to $\qrs \bigcup \qr$. 
By Observation \ref{obs:matching-cut} none of these crossing edges can be in $H$
because $M(H)$ is maximum with respect to $H$,
so there must be at least $\mu(G) - \mu(H)$ vertex-disjoint edges in $G \sm H$
between $\pls \bigcup \pl$ and $\qrs \bigcup \qr$.
Let $S$ contain the endpoints of these edges.
Clearly, $|S| = 2(\mu(G) - \mu(H))$, 
so property $2$ is satisfied.
Property $3$ is satisfied because $S$ is a subset of 
$\pls \bigcup \pl \bigcup \qrs \bigcup \qr$, 
and by Observation \ref{obs:matching-cut} 
all edges incident to that set go to $\pr \bigcup \ql = T$.

Finally, property $4$ is satisfied because by construction
$S$ has a perfect matching in $G \sm H$, 
and for each edge $(u,v)$ in the matching we know by property P2
of an \edcsab\ that $\dh(u) + \dh(v) \geq \be(1-\la)$;
since there are $\mu(G) - \mu(H)$ edges in the matching,
and $2(\mu(G) - \mu(H))$ vertices in $S$,
we conclude that the average degree $\dh(s)$ for $s \in S$ 
is at least $\be(1-\la)/2$.
\end{proof}

\noindent The intuition for the proof of Theorem \ref{thm:match-general} is as follows.
Let us say, for contradiction, that $\mu(H)$ was much smaller than $\mu(G)$. Then consider the sets $S$ and $T$ that exist according to 
Lemma \ref{lem:cut-sets-unweighted}. 
By property 4 of Lemma \ref{lem:cut-sets-unweighted}, 
each vertex $s \in S$ has an average of at least around $\be/2$ incident edges in $H$. But all the edges in $H$ incident to 
$S$ go to $T$, and $T$ has only $\mu(H)$, 
which is relatively small compared to $|S| = 2(\mu(G) - \mu(H))$
if $\mu(H)$ is much smaller than $\mu(G)$.
To close the contradiction we argue that because of property P1 of an \edcsab, we are simply not able fit all those edges from $S$ to $T$. We argue this by bounding how high degrees can get in an \edcsab. Intuitively, if $U$ and $V$ have equal size and all edges are between $U$ and $V$, we expect the average degree on each side to be no more than $\be / 2$, as if each vertex had degree $\be / 2$ then all edge degrees would be $\beta$ -- the maximum allowed by property P1. We now state a generalization of this intuition which shows that if one of the sets $U,V$ is larger than the other, it will have average degree below $\beta / 2$.

\begin{lemma}
\label{lem:unweighted-degree}
Let us say that in some graph we have disjoint sets $(U,V)$ such that $|U| = c|V|$, and all edges incident to $U$ go to $V$ (but there may be edges incident to $V$ which do not go to $U$). Let $d(v)$ be the degree of vertex $v$ in this graph, and say that for every edge $(u,v)$ in the graph $d(u) + d(v) \leq \De$ 
for some positive integer parameter $\De$. 
Then, the average degree of vertices in $U$ is at most $\frac{\De}{c+1}$. 
\end{lemma}

\begin{proof}
We want to upper bound the total number of edges from $U$ to $V$. Now,
there clearly exists a solution that maximizes this number in which
{\em all} edges in the graph are between $U$ and $V$ (i.e. no edges
from $V$ to elsewhere): just take a maximum solution that has other
edges and remove those -- the number of edges leaving $U$ remains the
same, and all constraints are clearly still satisfied. Thus, we can
assume for this proof that all edges in the graph are between $U$ and
$V$.

Let $E$ be the set of edges in the graph. Our goal is to upper bound
$|E| = \sum_{u \in U} d(u) = \sum_{v \in V} d(v)$. Now for each edge
$(u,v) \in E$ we have the constraint $d(u) + d(v) \leq \De$. Let us
sum the inequality constraints for all edges: this yields $\sum_{(u,v)
  \in E} d(u) + d(v) \leq |E|\De$. A closer look at the left hand side
shows that since each vertex $v$ appears in exactly $d(v)$ edges in
$E$, and each of those edges contributes $d(v)$ to the left hand side,
{\small
\begin{equation}
\label{eqn:upper-bound}
\sum_{(u,v) \in E} d(u) + d(v) = \sum_{u \in U} d(u)^2 + \sum_{v \in V} d(v)^2 \leq |E|\De .
\end{equation}}
Now that we have an upper bound, we also give a lower bound for
$\sum_{u \in U} d(u)^2$ and $\sum_{v \in V} d(v)^2$. Since we know
that $\sum_{v \in V} d(v)$ is fixed at $|E|$, the sum of squares is
minimized when all of the $d(v)$ are equal, i.e. when $d(v) = |E|/|V|$
for every $v$. The same is true for $U$, where recall that $|U| =
c|V|$. This yields:
{\small
\begin{equation}
\label{eqn:lower-bound}
\begin{split}
& \sum_{u \in U} d(u)^2 + \sum_{v \in V} d(v)^2 \geq 
\sum_{u \in U} \left( \frac{|E|}{|U|} \right)^2 + \sum_{v \in V} \left( \frac{|E|}{|V|} \right)^2 \\
&= \frac{|E|^2}{|U|} + \frac{|E|^2}{|V|} = \frac{|E|^2}{|V|}\cdot \left( 1 + \frac{1}{c}\right) = \frac{|E|^2}{|V|} \cdot \frac{1+c}{c}
\end{split}
\end{equation}}
\noindent Merging the upper bound from Equation \ref{eqn:upper-bound} and the lower bound from Equation \ref{eqn:lower-bound} we get that 
{\small $$ \frac{|E|^2}{|V|} \cdot \frac{1+c}{c} \leq |E|\De \Rightarrow |E| \leq \De|V| \cdot \frac{c}{1+c} .$$}
\noindent Thus, the average degree of $U$ is at most $\frac{|E|}{|U|} = \frac{|E|}{c|V|} \leq \frac{\De}{1 + c}$, as desired. 
\end{proof}

We note that bipartiteness was actually not required for the proof --
we only needed that all edges incident to $U$ go to $V$, which of
course disallows edges whose endpoints are both in $U$. 

\begin{proofof}{Theorem \ref{thm:match-general}}: Let us say, for the sake of contradiction, that we had $\mu(H) < (2/3 - \eps) \mu(G)$. Then, we have sets $S,T$ as in Lemma \ref{lem:cut-sets-unweighted}. 
By property 4 of Lemma \ref{lem:cut-sets-unweighted} the average degree
$\dh(s)$ among vertices $s \in S$ is at least $\beta(1-\la)/2$.
We argue such a high average degree is not possible. 
Since $\mu(H) < (2/3 - \eps) \mu(G)$:
{\small 
\begin{equation}
\label{eqn:free-size}
|S| = 2(\mu(G) - \mu(H)) > \mu(H)(1 + \eps) \ .
\end{equation}}
\noindent Observe that we are now in the situation
described in Lemma \ref{lem:unweighted-degree}: $S$
corresponds to $U$, $T$ corresponds to $V$, and 
the $\be$ parameter of the \edcsab\ $H$ corresponds
to the $\De$ parameter in Lemma \ref{lem:unweighted-degree}.
Property
3 of Lemma \ref{lem:cut-sets-unweighted} precisely tells us that all edges incident to
$U$ go to $V$, as needed in Lemma \ref{lem:unweighted-degree}. We know
from properties 1 and 2 of Lemma \ref{lem:cut-sets-unweighted} that 
$|V| = \mu(H)$ and $|U| = 2(\mu(H) - \mu(H))$  
so by Equation \ref{eqn:free-size} we have 
$|U| = c|V|$ for some $c > \oeps$.  
Thus Lemma
\ref{lem:unweighted-degree} tells us that the average degree of $U$ is
at most $\beta/(1 + c) \leq \beta/(2 + \eps)$, which some simple
algebra shows is strictly less than $\beta(1 - \lambda)/2$ because we
set $\lambda = \eps / 4$. We have thus arrived at
a contradiction with property 4 of Lemma \ref{lem:cut-sets-unweighted}, 
so our original assumption that $\mu(H) < (2/3 - \eps) \mu(G)$ must be false.
\end{proofof}

\subsection{Small Arboricity Graphs}

We now turn to Theorem \ref{thm:match-small}.
The statement is very similar to Theorem \ref{thm:match-general}, but with two crucial
differences: we are now dealing with a {\em weighted} \edcsab\ $H$, and
the approximation we need to guarantee is $1 - \eps$ instead of $2/3 -
\eps$. (Note that Theorem \ref{thm:match-small} is true of general
graphs as well; we only use it for small arboricity graphs, however,
because a weighted \edcsab\ is difficult to maintain in general
graphs.) It may seem unintuitive that a weighted \edcsab\ contains a
better matching than an unweighted one since it will in fact have
fewer total edges to work with. To show why a weighted \edcsab\ is 
better, see Figure~\ref{fig:bad} for a simple example where an unweighted
\edcsab\ only contains a $(3/2)$-approximate matching, but a
weighted one does not suffer the same issues.

To prove Theorem \ref{thm:match-small} we show that if the \edcsab\ $H$ is 
a \emph{weighted} \edcsab\ (Definition \ref{dfn:weighted-edcs}),
then we can prove stronger versions of
Lemmas \ref{lem:cut-sets-unweighted} and \ref{lem:unweighted-degree}.
Before doing so, we prove a simple lemma regarding the degrees $\dh(v)$
of a certain sort of "alternating" path which will end up corresponding
to an augmenting path in $G$ with respect to a maximum matching in $H$

\begin{lemma}
\label{lem:path-degree}
Let $G = (V,E_G)$ be a bipartite graph, and let $H = (V,E_H)$ be a
a \emph{weighted} \edcsab\ with $\bem = \be - 1$.
Let $P$ be some path in $G$ with endpoints $s$ and $t$. 
Let $L(P)$ be the number of vertices in $P$
and say that $P$ has the following property: 
$L(P)$ is even,
and every even edge in $P$ (the second edge, the fourth edge, and so on) is in $H$. 
(The odd numbered edges can be in either $H$ or $G \sm H$).
Then: 
$$\dh(s) + \dh(t) \geq \be - 1 - \frac{L(P) - 2}{2} \ .$$ 
\end{lemma}

\begin{proof} 
Say that the vertices on $P$ are $s = s_1, s_2, s_3, ..., s_{L(P)} = t$. 
We will prove by induction that for any even index $k$,
$$\dh(s_1) + \dh(s_k) \geq \be - 1 - \frac{k-2}{2} \ .$$
Setting $k = L(P)$ then yields the statement of the theorem.

For the base case, when $k = 2$, then by property P2 of a weighted \edcsab\
since the edge $(s_1, s_2)$ is in $G$ we have 
$\dh(s_1) + \dh(s_k) \geq \be - 1$, as desired.
Now, say that the statement is true for some even $k$.
We want to prove that it is also true of $k+2$.
We know that the edge $(s_k, s_{k+1})$ is an even edge in $P$,
so by the assumption of the theorem it is in $H$.
Thus, by property P1 of a weighted \edcsab\ we have 
$\dh(s_k) + \dh(s_{k+1}) \leq \be$.
But by property P2 of a weighted \edcsab\ we have
$\dh(s_{k+1}) + \dh(s_{k+2}) \geq \be - 1$.
Subtracting the former inequality from the latter we get
$\dh(s_{k+2}) - \dh(s_k) \geq -1$.
Now, we want to lower bound
$$\dh(s_1) + \dh(s_{k+2}) = \dh(s_1) + \dh(s_k) 
+ (\dh(s_{k+2}) - \dh(s_k)) \geq \dh(s_1) + \dh(s_k) - 1 \ .$$
The induction hypothesis then yields the desired result:
$$\dh(s_1) + \dh(s_{k+2}) \geq \dh(s_1) + \dh(s_k) - 1
\geq \be - 1 - \frac{k - 2}{2} - 1 
= \be - 1 - \frac{k + 2 - 2}{2} \ .$$
\end{proof}

\begin{lemma}
\label{lem:cut-sets-weighted}
Let $G = (V,E_G)$ be a bipartite graph, and let $H = (V,E_H)$ be a
\emph{weighted} \edcsab\ with $\bem = \be - 1$.
Then, there exist disjoint sets of vertices $S$, $T$, $\tstar$ with the following properties.
\begin{enumerate}
\item $|T| = |\tstar| = \mu(H)$ 
and there is a perfect matching in $H$ between $T$ and $\tstar$
\item $|S| = 2(\mu(G) - \mu(H))$
\item All edges incident to $S \bigcup \tstar$ go to $T$.
\item The average degree of $\dh(s)$ of vertices in $S$ is at least
$$\frac{\beta - 1}{2} - \frac{\mu(H)}{4(\mu(G) - \mu(H))} \ .$$ 
\end{enumerate}
\end{lemma}

\begin{proof}
Let $M(H)$ be some maximum matching in $H$.
$M(H)$ induces a $(P,Q)$ cut in the graph $H$ as defined in 
Definition \ref{dfn:matching-cut}, 
which partitions the vertices into sets
Let $\pls, \pl, \pr, \qrs, \ql, \qr$.
We set $S = \pls \bigcup \qrs$, $T = \pr \bigcup \ql$ and $\tstar = \pl \bigcup \qr$.
It is easy to see from Observation \ref{obs:matching-cut} that these sets
satisfy the first three properties of the lemma to be proved.

To prove the fourth property, 
we start by observing that there must be $\mu(G) - \mu(H)$ 
vertex-disjoint augmenting paths from $\pls$ to $\qrs$. 
Let $\Pi$ be the set of these augmenting paths,
and for each path $P \in \Pi$ let $L(P)$ denote the number of vertices in $P$.
Observe that every $P \in \Pi$ satisfies the properties necessary for Lemma 
\ref{lem:path-degree} to hold:
$L(P)$ is even because $P$ goes from the left side of the bipartition to the right,
and every even edge in $P$ is in the matching $M(H)$ so in particular it is in $H$.
Now each of the $2(\mu(G) - \mu(H))$ vertices in $S = \pls \bigcup \qrs$ is an endpoint
of one of the $\mu(G) - \mu(H)$ augmenting paths $P \in \Pi$, 
so $\sum_{s \in S} \dh(s)$ is equal to 
$\sum_{u} \dh(u)$ over all vertices $u$ that are endpoints of one of the $P \in \Pi$.
Thus by Lemma \ref{lem:path-degree} we have:
\begin{equation}
\label{eqn:S-degree}
\sum_{s \in S} \dh(s) \geq (\be - 1)(\mu(G) - \mu(H)) 
- \frac{1}{2} \sum_{P \in \Pi} (L(P) - 2) \ .
\end{equation}
But note that for every augmenting path $P \in \Pi$, 
all the vertices in $P$ except the two endpoints are matched in $H$
(by definition of an augmenting path),
so $(L(P) - 2)$ is the number of matched vertices in $P$.
Since all the paths $P \in \Pi$ are disjoint we have that 
\begin{equation}
\label{eqn:P-sum}
\sum_{P \in \Pi} (L(P) - 2) \leq \mu(H) \ .
\end{equation}
Thus, combining equations \ref{eqn:S-degree} and \ref{eqn:P-sum} we get
$$\sum_{s \in S} \dh(s) \geq (\be - 1)(\mu(G) - \mu(H)) - \frac{\mu(H)}{2} \ .$$
Diving this by $|S| = 2(\mu(G) - \mu(H))$ yields the 
\emph{average} $\dh(s)$ among $s \in S$ stated in the lemma.
\end{proof}

Lemma \ref{lem:cut-sets-weighted} is 
similar to Lemma \ref{lem:cut-sets-unweighted},
except that it also guarantees the existence of a set $\tstar$
which can be perfectly matched to $T$ using edges in $H$.
(Lemma \ref{lem:cut-sets-weighted} also contains a slightly weaker
lower bound on the average degree in $S$, 
but this ends up having only a small impact on the final result).
In the proof of Theorem \ref{thm:match-general}
we took the sets $S,T$ from Lemma \ref{lem:cut-sets-unweighted}
and argued that since $S$ has high average degree and all edges
from $S$ go to $T$, in order for all those edges to fit into $T$,
the set $T$ itself has to be relatively large compared to $S$.
Now, to prove Theorem \ref{lem:cut-sets-weighted} we 
need to show even stronger bounds on the size of $T$ relative to $S$.
We do this by arguing that not only must $T$ be able to fit all the edges coming from $S$,
it must also be able to fit all the weight coming from $\tstar$ to $T$.
Note that even if there are not a large number of edges from $\tstar$ to $T$
(we only guarantee a single perfect matching worth of edges),
we are dealing with a \emph{weighted} \edcsab, 
so the edges could have high weight.
We now formalize this intuition by proving a generalization of 
Lemma \ref{lem:unweighted-degree}.

\begin{lemma}
\label{lem:weighted-degree}
Say that in some graph we have two disjoint sets $U,V$ such that all
edges incident to $U$ go to $V$. Let $V = \{v_1, ..., v_n\}$, and let
$U = W \bigcup X$, where $W = \{w_1, ..., w_n\}$ and $X = \{X_1, ...,
X_{cn}\}$ for some $c < 1$. Note that $|W| = |V|$, $|X| = c|V|$ and
$|U| = (1 + c)|V|$. Now, say that all edges have positive integer
weights and that the degree of vertex $v$ (denoted $d(v)$) is the sum
of its incident edge weights. Say also that the graph obeys the
following degree constraints, for some positive integer parameter $\De$:
\begin{itemize}
\item {\bf Constraint 1:} for every edge $(u,v)$ between $U$ and $V$ we have $d(u) + d(v) \leq \De$ (Compare this with property P1 of a weighted \edcsab).
\item {\bf Constraint 2:} for all $n$ pairs $(v_i, w_i)$, we have $d(v_i) + d(w_i) \geq \De -1$.  (Compare this with property P2 of a weighted \edcsab.)
\end{itemize}
\noindent Then, the average degree in $X$ is at most $\De / (2+c) + 1/c$ (so around $\De / (2+c)$ for large enough $\De$).

\end{lemma}

\begin{proof}
We want to upper bound the number of edges incident to $X$. We will start by arguing that there is some graph that maximizes this quantity where {\em all} edges of the graph are between $U$ and $V$. Let us start with some valid graph that maximizes the number of edges leaving $X$, but might also have other edges. We will show that we can always remove any edge that is not between $U$ and $V$ and then fix up the graph in such a way that none of the degrees in $X$ decrease but all the constraints are still satisfied: repeating this multiple times, we will end up with a graph where all edges are between $U$and $V$ but the total degree of $X$ is still maximized.

The two constraints above only concern degrees in $U$ and $V$, so clearly any edge that is incident to neither $U$ nor $V$ can be safely removed. Now let us take some edge $(*,v_i)$ that is incident to $V$ but not $U$. Removing this edge decreases the degree of $v_i$, which might violate constraint 1 concerning $(v_i, w_i)$; thus, we might now have some fixing up to do. To do this, let us define a vertex $v_i \in V$ to be {\em deficient} if $d(v_i) + d(w_i) = \De - 1$. Let us define edge $(u,v)$ to be {\em full }if $d(u) + d(v) = \De$. Notice that we can safely raise the degree of any vertex that has no incident full edges without violating any of the constraints; similarly, we can decrease the degree of any vertex $v_i \in V$ that is not deficient. 

Now, once we remove edge $(*,v_i)$ (the edge not between $U$ and $V$)
the degree of $v_i$ is about to decrease. If $v_i$ is not deficient we
allow this to happen and we are done. Otherwise, we add a single unit
of weight to edge $(v_i, w_i)$; note that if the edge doesn't exist we
can simply add it to the graph, since we only need to prove that there
exists \emph{some} solution that maximizes the total degree of $X$
while only using $U$-$V$ edges. The degree of $v_i$ thus remains
unchanged, but the degree of $w_i$ is about to increase. If $w_i$ has
no incident full edges we allow this to  happen and we are
done. Otherwise, let $(w_i, v_{i_2})$ be one of these full edges, and
decrease its weight by 1. The degree of $w_i$ thus remains unchanged
but the degree of $v_{i_2}$ is about the decrease. We now repeat: if
$v_{i_2}$ is not deficient we allow its degree to decrease and we are
done; else, we add one unit of weight to $(v_{i_2}, w_{i_2})$. If
$w_{i_2}$ has no incident full edges we allow its degree to increase
and we are done; otherwise we remove one unit of weight from $w_{i_2},
v_{i_3}$. As we continue in this fashion, we are always ensuring that
all constraints are satisfied. It is also easy to see that no degrees
in $U$ decrease: all of them remain the same (every weight-decrease is
preceded by a weight-increase), except for the last vertex examined
which might {\em increase} its degree by 1. Thus, all we have left to
show is that this fixing up process terminates. We show this by
proving that $d(w_{i_k})$ is always strictly smaller than
$d(w_{i_{k+1}})$. Since the algorithm didn't stop at $d(w_{i_k})$ it
must have found a full edge $(w_{i_k}, v_{i_{k+1}})$, so by definition
of full $d(v_{i_{k+1}}) = \De - d(w_{i_k})$. But since the algorithm
didn't stop at $d(v_{i_{k+1}})$ it must have been deficient, so
$d(w_{i_{k+1}}) = \De -1 - d(v_{i_{k+1}}) = \De - 1 - (\De -
d(w_{i_k})) = d(w_{i_k}) - 1$ (This argument is analogous to one used
in Lemma \ref{lem:path-length} -- see Figure~\ref{fig:augpath}).

{\em Thus we can assume for the rest of the proof that all edges in the graph are between $U$ and $V$}. This implies that the total degree of $U$ is equal to the total degree of $V$. We now use Lemma \ref{lem:unweighted-degree} to bound this total degree (recall that in our setup for this lemma, $|V| = |W| = n$, $|X| = cn$, and $|U| = (1+c)n$).

$$\sum_{v \in V} d(v) = \sum_{u \in U} d(u) \leq \frac{\De|U|}{1 + (1 +c)}=\frac{\De|U|}{2 + c} = \De n \frac{1 + c}{2+c}.$$

\noindent Now, constraint 2 of our lemma clearly implies that 
$$\sum_{w \in W} d(w) \geq n(\De - 1) - \sum_{v \in V} d(v) = n(\De - 1) - \sum_{u \in U} d(u).$$

 \noindent But note that since $U = W \bigcup X$ we have  
\begin{equation*}
\begin{split}
\sum_{x \in X} d(x) &= \sum_{u \in U} d(u) - \sum_{w \in W} d(w) \leq 2\sum_{u \in U} d(u) - n(\De -1) \\
											&\leq 2\De n \frac{1 + c}{2+c} - n\De + n = \De n \frac{c}{2+c} + n
\end{split}
\end{equation*}

Dividing this through by $|X| = cn$ we get that the average degree in $|X|$ is at most $\De / (2+c) + 1/c$.
\end{proof}

As in Lemma \ref{lem:unweighted-degree}, bipartiteness is not required here; we only need that all edges incident to $U$ go to $V$. 

\begin{proofof}{Theorem \ref{thm:match-small}}
The proof is very similar to that of Theorem \ref{thm:match-general}. Say for contradiction that $\mu(H) < (1 - \eps)\mu(G)$. 
Consider the sets $S,T,\tstar$ guaranteed by Lemma \ref{lem:cut-sets-weighted}.
By Property 4 of Lemma \ref{lem:cut-sets-weighted} 
the average degree $\dh(s)$ of $s \in S$ is at least 
$$\mbox{average of $\dh(s)$} 
\geq \frac{\beta - 1}{2} - \frac{\mu(H)}{4(\mu(G) - \mu(H))} \ .$$ 
Combining this with the contradiction-assumption that $\mu(H) < (1 - \eps)\mu(G)$,
we have 
\begin{equation}
\label{eqn:average-lower-unweighted}
\mbox{average of $\dh(s)$} 
\geq \frac{\beta - 1}{2} - \frac{\mu(H)}{4(\mu(G)-\mu(H))} 
\geq \frac{\beta - 1}{2} - \frac{1}{4\eps} \ .
\end{equation}
 
We now use Lemma \ref{lem:weighted-degree} to show that such
a high average degree is not possible, thus yielding the desired
contradiction. To invoke Lemma \ref{lem:weighted-degree},
we let $X = S$, $W = \tstar$ (so $U = S \bigcup \tstar$), and $V = T$;
the edges of the graph are the edges of $H$,
and we set the parameter $\De$ in Lemma \ref{lem:weighted-degree}
to be $\beta$ (the \edcsab\ parameter).

Property 3 of Lemma \ref{lem:cut-sets-weighted} guarantees that
all edges in $U = S \bigcup \tstar$ go to $V = T$,
as needed in Lemma \ref{lem:weighted-degree}.
Constraint 1 of Lemma \ref{lem:weighted-degree} is satisfied because 
of property P1 of a weighted \edcsab\, 
and constraint 2 of Lemma \ref{lem:weighted-degree} is satisfied
because there is a perfect matching in $G$ between $\tstar$ and $T$,
and for each of these edges property P2 of a weighted \edcsab\ holds.
Now, since we assumed for contradiction that $\mu(H) < (1 - \eps)\mu(G)$,
we have that $|S| = 2(\mu(G) - \mu(H)) > 2\eps\mu(H) = 2\eps|T|$,
so $|S| = c|T|$ for some $c > 2\eps$. 
Thus by Lemma \ref{lem:weighted-degree} we have:
\begin{equation}
\label{eqn:average-upper-unweighted}
\mbox{average of $\dh(s)$} < \frac{\beta}{2+2\eps} + \frac{1}{2\eps} \ .
\end{equation}

Some simple algebra now shows that Equations 
\ref{eqn:average-lower-unweighted} and 
\ref{eqn:average-upper-unweighted}
are contradictory because together they imply that 
$$\frac{\beta}{2+2\eps} + \frac{1}{2\eps} > \frac{\be - 1}{2} - \frac{1}{4\eps} \ .$$
But the statement of Theorem \ref{thm:match-small} assumes that $\be > 8\eps^{-2}$
and that $0 < \eps < 1$ so we can also show the opposite to be true, i.e. that:
\begin{equation}
\begin{split}
\frac{\beta}{2+2\eps} + \frac{1}{2\eps} 
&< \frac{\be}{2}(1 - \frac{\eps}{2}) + \frac{1}{2\eps} \\
&= \frac{\be}{2} - \frac{\be\eps}{4} + \frac{1}{2\eps} \\
&< \frac{\be}{2} - \frac{2}{\eps} + \frac{1}{2\eps} \\
&= \frac{\be}{2} - \frac{3}{2\eps} \\
&< \frac{\be}{2} - 1 - \frac{1}{2\eps} \\
&< \frac{\be-1}{2} - \frac{1}{4\eps}
\end{split}
\end{equation}
We have thus reached a contradiction, so our original assumption that 
$\mu(H) < (1 - \eps) \mu(G)$ must be false, 
which proves Theorem \ref{thm:match-small}.
\end{proofof}

\negsp
\section{Maintaining an \edcs}
\label{sec:edcs}
\negsp

In this section, we outline the proofs of 
Theorems~\ref{thm:h-general} and~\ref{thm:h-small}. Due to space constraints,
we leave the formal proof for Section~\ref{sec:ap-h}.

Recall that $\ed(u,v)$ denotes the edge degree of $(u,v)$, $\dh(u) + \dh(v)$. 
We define an edge to be \emph{full} if it is in $H$ and has edge degree $\be$. We define it to
be \emph{deficient} if it is not in $H$ and has the minimum allowable edge degree $\bem$, which is $\beta - 1$ for 
the weighted \edcsab\ in Theorem \ref{thm:h-small} and $\be(1 - \la)$ for the unweighted \edcsab\ of Theorem \ref{thm:h-general}.
We define a vertex to be \emph{increase-safe} if it has no incident full edges and \emph{decrease-safe} if it has no incident deficient edges;
it is easy to see that increasing (decreasing) the degree of an increase-safe (decrease-safe) vertex by one does not lead to a violation of any \edcsab\ constraints.

Now, let us say that we delete some edge $(u,v)$ from $G$. If $(u,v)$
was not in the \edcsab\ $H$ then all constraints remain
satisfied. Otherwise, deleting $(u,v)$ causes the degrees of $u$ and
$v$ to decrease by one. Let us focus on fixing up vertex $v$; vertex
$u$ can then be handled analogously.  If $v$ was decrease-safe, then
all constraints relating to $v$ remain satisfied and we are done.
Otherwise, it must have had some incident deficient edge $(v,
v_2)$. Adding this edge to $H$ rebalances the degree of $v$ to what it
was before the deletion, but now the degree of $v_2$ has increased by
one. If $v_2$ was increase-safe, the degree increase does not violate
any constraints, and we are done. Otherwise, $v_2$ must have an
incident full edge $(v_2, v_3)$ which we delete from the graph; this
rebalances $v_2$ but decreases the degree of $v_3$, so we look for an
incident deficient edge. We continue in this fashion until we end on
an increase/decrease-safe vertex.

We can thus fix up an edge deletion by finding an alternating path of
full and deficient edges that ends in an increase/decrease-safe
vertex.  Insertions are handled analogously.  This process is similar
to finding an augmenting path in a matching except that finding an
augmenting path is much harder because we might hit a dead end and
have to back track; but we can fix up an \edcsab\ by following
\emph{any} sequence of full/deficient edges.  Moreover, the resulting
alternating path is always simple and contains few edges.  
Figure~\ref{fig:augpath} illustrates this point.  For the small arboricity case
(Theorem \ref{thm:h-small}) where $\bem = \be - 1$, it is not hard to
see that in any such alternating path the vertex degrees $\dh(v)$ on
either side of the bipartition are either increasing or decreasing by
1, so since $\dh(v)$ is always between $0$ and $\beta$, the path has
length $O(\beta)$.  In this small arboricity case, $O(\beta)$ is small
because $\beta = O(1/\eps^2)$ (See Section
\ref{sec:proof-main-small}).  In the general case (Theorem
\ref{thm:h-general}), $\beta$ is large but the gap between $\beta$ and
$\bem$ is $\be \la$, so degrees on either side change by $\be \la$ and
the path has length only $O(1/\la)$.

To find such an alternating path of full and deficient edges we maintain
a data structure that for any vertex $v$ can return an incident full or deficient edge 
(whichever is asked for), or indicate that none exists. 
Since the alternating path will always be short, this data structure will only be queried
a small number of times per insertion/deletion in $G$. 
We maintain this data structure using a dynamic orientation,
in which each edge is owned by one of its endpoints (see end of Section \ref{sec:preliminaries}).
Let us focus on the small arboricity case, where the dynamic orientation maintains a small max load.
Each vertex will maintain fullness/deficiency information about the edges it does \emph{not} own, 
storing each category of edge (full/deficient) in its own list. 
To find a full/deficient edge incident to some vertex $v$, 
the data structure simply picks an edge from the corresponding list in $O(1)$ time;
if the list is empty, the data structure then manually checks all the edges that $v$ \emph{does} own: 
since the max load is small, this can be done efficiently. 
When the status of a vertex $v$ changes, to maintain itself the data structure must  
transfer this information along all edges $(v,u)$ that are \emph{not} owned by $u$,
but since these are precisely the edges owned by $v$, 
there can only be a small number of them.

The basic idea is the same for general bipartite graphs (Theorem \ref{thm:h-general}), 
except that now the max load is $O(\sqrt{m})$, and 
we cannot afford to spend $O(\sqrt{m})$ per update. Note that in this case, however, there is a gap of $\be \la$ 
between full and deficient edges, so intuitively, the degree of a vertex has to change $\be \la$ time before 
it must be updated in the data structure. This leads to an update time of around $\sqrt{m}/(\be \la)$, as
needed in Theorem \ref{thm:h-general}. The details, however, are quite involved, especially since
we need a \emph{worst-case} update time. 

\negsp
\section{Conclusion}
\label{sec:conclusions}
We have presented the first fully dynamic bipartite matching algorithm to
achieve a $o(\sqrt{m})$ update time while maintaining a
better-than-2-approximate matching. It is also the fastest known
deterministic algorithm for achieving {\em any} constant
approximation, and certainly any better-than-2 approximation. The main
open questions are in how far we can push this tradeoff. Can we
achieve a {\em randomized} better-than-2 approximation with
update time polylog(n)? For {\em deterministic} algorithms, can we 
achieve a constant approximation with update time polylog(n), 
or  a $(\oeps)$-approximation with update time
$o(\sqrt{m})$?

The other natural question is whether our results can be extended to
general (non-bipartite) graphs and non-bipartite graphs of small
arboricity. The definition of an \edcs\ does not inherently rely on
bipartiteness, and neither do many of the techniques in this paper.
The main obstruction to the generalization seems to lie in 
the standard cut induced by a matching in bipartite (and only bipartite) graphs
(Definition \ref{dfn:matching-cut}), 
which was crucial to proving Lemmas 
\ref{lem:cut-sets-unweighted} and 
\ref{lem:cut-sets-weighted}. 

\section{Acknowledgements} 
We would like to thank Tsvi Kopelowitz for very helpful pointers about dynamic orientations. We would also like to thank Virginia Williams for pointing out
the mistake in Lemma 2 of the conference version of this paper.
(See the remark at the end of Section \ref{sec:results}).


\bibliographystyle{plain}
\bibliography{dynamic_matching}

\newpage

\appendix

\section{Proof of Theorem \ref{thm:main-small}}
\label{sec:proof-main-small}
The proof follows from Theorems \ref{thm:match-small} and \ref{thm:h-small}, 
and is analogous to the proof of Theorem \ref{thm:main-general} given in Section \ref{sec:framework-general}.

We use the algorithm outline presented near the beginning of Section \ref{sec:framework}.
For our transition subgraph $H$, we use a weighted \edcsab($G, \beta, \beta-1$) with $\beta = 8\eps^{-2}$. 
By Theorem \ref{thm:h-small} we can maintain $H$ in worst-case update time 
$O(\beta^2(\alpha + \log(n)) + \al(\alpha + \log(n))) = O(\eps^{-4}(\alpha + \log(n)) + \al(\alpha + \log(n)))$.
The update ratio of the algorithm is $O(\be)$.
This \edcsab\ clearly has max degree $\beta$ so 
by Lemma \ref{lem:maxdegree} we can then maintain a $\oeps$-approximate matching in $H$ in time
$O(\beta\eps^{-2})$; multiplying by the update ratio $O(\beta)$ of maintaining $H$ in $G$, 
we need $O(\beta^2\eps^{-2}) = O(\eps^{-6})$ time to 
maintain the matching per change in $G$. 
Combining the terms above gives precisely the bound of Theorem \ref{thm:main-small}.
By Theorem \ref{thm:match-small}, $\mu(H)$ is a $\oeps$-approximation to $\mu(G)$, so 
our matching is a $\oeps \oeps = \oeps$-approximate matching in $G$, as desired.

\section{Proof of Theorem~\ref{thm:match-small}} 
\label{sec:app-proof-match-small}

In this section we give a full proof of Theorem \ref{thm:match-small}. Recall the intuition
given at the end of Section \ref{sec:matching}; in particular, we will rely
on the following generalization of Lemma \ref{lem:unweighted-degree}.

\section{Dynamically Maintaining an EDCS in a Bipartite Graph}
\label{sec:ap-h}

In this section we provide a formal proof of Theorems \ref{thm:h-general} and \ref{thm:h-small}. 
We address the low arboricity case first (Theorem \ref{thm:h-small}), as it is a simpler
algorithm and analysis, and then explain how to extend our work to
general graphs. 

In both cases, we will show that when an edge is inserted or deleted
in $G$, we only need to do a small number of updates to maintain an
$H$ with the desired \edcsab\ properties.  We will show that we will always be
able to find a specific type of alternating path that will allow us
to maintain $H$.  We will need to show that the length of the
alternating path is bounded, and that we can also find such a path
efficiently.  Finding a path will involve using the edge orientation
to control exactly which neighbors need to be notified of a change in
vertex degree. In the general graph case, we will not be able to
efficiently maintain accurate degree counts, so we will only maintain
approximate counts and use a bucketing scheme to identify edges of
appropriate degree.

\subsection{Dynamically Maintaining an EDCS in Small Arboricity Graphs}
\label{sec:ap-h-small}
For small arboricity graphs, we will ultimately maintain a
{\em weighted} \edcsab.  We first describe how to maintain a {\em unweighted}
\edcsab\  with $\bem = \be-1$
and then at the end of this section we extend the result to a weighted
\edcsab\ with $\bem = \be-1$. To maintain the unweighted \edcsab\  we define two classes of edges:

\begin{itemize}
\item a {\em full} edge $(u,v)$ is in $H$ and has $\dh(u) + \dh(v) = \beta$
\item a {\em deficient} edge $(u,v)$ has $\dh(u) + \dh(v) = \beta - 1$.
\end{itemize}

Recall that when an edge is inserted or deleted, we first update the
orientation, thereby causing some number of edge flips. Recall that a vertex always
accurately knows its own degree, but may not accurately know the
degree of all its neighbors.  From the orientation, each vertex owns some edges.  We will
maintain the invariant that we always know accurately the degree of our unowned incident edges.   Also recall that we use $\ed(u,v)$ to denote the edge degree of $(u,v)$, $\dh(u) + \dh(v)$.

\paragraph{Insertion}
Consider an edge $(u,v)$ that has just been added to $G$.    If
$\ed(v,w) \geq \beta -1$, then we do not add the edge to $H$, and the 
properties P1 and P2 from Definition~\ref{dfn:unweighted-edcs}
remain satisfied.  On the other hand, if 
$\ed(v,w) < \beta -1$, we want to  add $(u,v)$ to $H$.  Doing so 
will increase $d(u)$ and $d(v)$ by 1, which may lead to a violation of
$P1$ for other edges in $H$ that are incident to either $u$ or $v$.  
Thus, we will need to find a type of  alternating path that will allow us to
add $(u,v)$ and still maintain P1 and P2 for all vertices.

We say that a vertex $x$ is {\em increase-safe} if it has no incident
full edges, and say that it is {\em decrease-safe} if it has no
incident deficient edges.  Returning to adding $(u,v)$ to $H$, let's
focus on vertex $v$; we will then deal with vertex $u$ analogously.  If $v$ is increase safe, then when we add
$(u,v)$ we have not violated P1 for any edges incident to $v$, since
there are no full edges incident to $v$.  If $v$ is not increase safe,
then it must have at least one  incident full edge, say $(v, p_1)$.
We would like to add $(u,v)$ to $H$ and remove $(v,p_1)$ from $H$,
thereby leaving $v$'s degree unchanged.  Doing so would decrease
$d_H(p_1)$, which we can do only if $p_1$ is decrease safe.  If $p_1$
is decrease safe, then adding $(u,v)$ to $H$ and removing $(v,p_1)$
leaves $v$'s degree unchanged, decreases $p_1$'s degree and
reestablishes $P1$ and $P2$ for all vertices (except possibly $u$).
However, if $p_1$ is not decrease-safe, it must have an incident
deficient edge, say $(p_1,p_2)$.  We can add this edge to $P$
and continue from $p_2$ we did from $v$.   We can continue in this
manner, stopping when we find either an increase-safe vertex or decrease-safe vertex.
Assume that the set of edges we find form a simple path.  Then we can exchange the role of the matched
and unmatched edges on $P$
thereby reestablishing $P1$ and $P2$ for all
vertices (except possibly $u$).  Note that this path may leave the
number of edges in $H$ unchanged, or may increase the number of edges
in $H$ by one.  Either outcome is acceptable.

We now argue that the set of edges we find do form a simple path.
In addition, in order to bound the time, we would
like to bound the length of $P$.  
We do so with the following lemma:

\begin{lemma}
\label{lem:path-length}
Let $P$ be a path of alternating full and deficient edges.  Then $P$
is simple and the length of $P$ is at most $ 2 \beta +1 $.
\end{lemma}

\begin{proof}
Consider first the case that $P = (p_0,p_1,\ldots,..,p_k)$ starts with
a full edge.  Let $d = d_H(p_0)$ and clearly $d\leq \beta$.  Since
$(p_0,p_1)$ is full, $d_H(p_1) = \beta - d$.  Since $(p_1,p_2)$ is
deficient, 
$d_H(p_2) = \beta - d_H(p_1) - 1 = \beta - (\beta -d ) -1 = d - 1$.
Continuing, we get that $d_H(p_3) = \beta -d + 1$, $d_H(p_4) = d-2$,
and in general $d_H(p_i) = d- 2*i$ for 
even $i$ (See Figure~\ref{fig:augpath}).  Since each vertex in $P$  has an
incident full edge, all vertices have positive degree, and thus $P$
can have at most $2 \beta$ edges.  Furthermore, if we consider all the
vertices on $P$ that are on the same side of the bipartite graph, they
all have distinct $\dh$ values, and hence they must be distinct and the
path is therefore simple.

If $P$ starts with a deficient
edge,   we can go through the same argument.  Now the degrees
of the odd indexed vertices are decreasing, and we have that the
length of the path is at most $2 \beta + 1$.
\end{proof}

Note that these alternating paths are analogous to augmenting paths in a standard B-matching, but are much more locally well behaved: when searching for an ordinary augmenting path we might reach a dead end and have to backtrack, but in an \edcsab\ following {\em any} full/deficient edge is guaranteed to eventually lead towards the desired alternating path.
After finding an alternating path from $v$, we repeat the same
procedure starting at $u$.  (Note that it is fine for the path from
$v$ to intersect the path from $u$, as we execute the fixing up
procedures sequentially). We have thus shown the following:

\begin{lemma}
\label{lem:insertion-length}
After inserting an edge into G, we can reestablish P1 and P2 using at most $4 \beta$ insertions/deletions
from $H$.
\end{lemma}

\paragraph{Deletions}

Deleting an edge $u$ from $G$ is handled in a similar manner to
insertions.  If $(u,v)$ is not in $H$, then we do not need to change
$H$.  If $(u,v)$ is in $H$ and both $u$ and $v$ are
decrease-safe, we just remove the edge $(u,v)$.  Otherwise, we find an
alternating path in the same way we did for insertions and observe
that Lemma~\ref{lem:path-length} applies for paths beginning with both
full and deficient edges.  Thus we have:
 
\begin{lemma}
\label{lem:deletion-length}
After deleting an edge from G, we can reestablish P1 and P2 using at most $4 \beta$ insertions/deletions
from $H$.
\end{lemma}

\paragraph{Finding alternating paths}

In order to find the alternating paths, we need to maintain the
necessary data structures to identify full and deficient edges.  Each
vertex $v$ will maintain the following information: 1)
$d_H(v)$, its degree in $H$, 2) a set $\owned(v)$ consisting of the edges it owns, 
3) a set  $\full(v)$ consisting of its unowned incident full incident edges,
and 4), a set $\defic(v)$ consisting of it unowned incident deficient edges.

\noindent Each of these sets has no particular order, and can be maintained
easily as a doubly linked list.   

We now conclude this section and provide a proof of
Theorem~\ref{thm:h-small}.  In the proof, we will also explain how to
maintain a weighted \edcs\ rather than an unweighted one.

\begin{proofof}{Theorem~\ref{thm:h-small}} 
By Lemmas~\ref{lem:insertion-length}
and~\ref{lem:deletion-length}, the update ratio (see Definition \ref{dfn:update-ratio}) is clearly
$O(\beta)$.
To bound the update time, we first perform $O(\al + \log(n))$ reorientations 
using Theorem \ref{thm:dynamic-orientation-arb}, which
takes $O(\al(\alpha + \log n))$ time.  For each flipped edge
$(v,w)$  we update the vertices $v$ and $w$, moving the edge in/out of the
lists $\owned()$, $\full()$ and $\defic()$ as appropriate. This takes $O(1)$ per flip, so $O(\al + \log(n))$ time in total.

Next we need to 
implement the search for the alternating path $P$ of full/deficient edges.
At vertex $v$, to search for an incident full edge, just check the
set $\full(v)$.  If it is non-empty, a full edge is found.  If it is
empty, then check the owned edges.  There are only  $O(\alpha + \log
n)$  owned edges, so this operation takes $O(\alpha + \log n)$
time; by Lemmas \ref{lem:insertion-length} and \ref{lem:deletion-length}, this process is repeated $O(\beta)$ times for a total of $O(\beta(\alpha + \log n))$ time.

Once we find an alternating path, we exchange its matched/unmatched edges to preserve properties P1 and P2;
this can change the degrees of at most two vertices (the path's endpoints) and so change the fullness/deficiency of their
edges. The owned neighbors of these vertices may thus have to modify their sets E() and F(), but each vertex owns at most 
$O(\al + \log(n))$ edges, so this takes $O(\al + \log(n))$ time in total.

We next  extend our algorithm to a weighted \edcsab\  by paying an extra factor of
$O(\beta)$  in the running time, thinking of the an edge of weight $w$
in the weighted \edcsab\ as
$w$ parallel edges in the unweighted one.  Observe that all weights
are bounded by $\beta$. The only change to the algorithm is 
the implementation of an edge deletion.  Now, if an edge is deleted
from $G$, it may have weight up to $\beta$ in $H$.  However, we can
simply delete from $H$ all $\beta$ (unweighted) parallel edges, using the algorithm for an
unweighted \edcsab.  This will increase the running time by a factor of
at most $\beta$.

All together the time to process an insertion/deletion in $G$ is
$O(\beta^2(\alpha + \log n) + \al(\alpha + \log n))$.
\end{proofof}

\subsection{Dynamically Maintaining an EDCS in General Bipartite Graphs}
\label{sec:ap-h-general}
In this section, we describe how to maintain $H$ in a general
bipartite graph (Theorem \ref{thm:h-general}).  At a high level, we use similar ideas to the
general case -- we will use an orientation to describe a data
structure consisting of owned and unowned edges and we will, when
edges are inserted or deleted, look for alternating paths of full and
deficient edges.  There will, however, be several technical
differences.  First, we will maintain an unweighted \edcs.  The
biggest difference however, is that, when we orient edges, by
Theorem~\ref{thm:dynamic-orientation-general}, a vertex may own up to 
$3 \sqrt{m}$ edges.  Therefore, when a vertex degree changes, we do not
have time to update all $3 \sqrt{m}$ neighbors, we will only have time
to update a small fraction of them.  Thus, we will not be able to assume that
we accurately know the degrees of our neighbors, and therefore know
which edges are full and which are deficient.  To compensate for this
lack of knowledge, we will introduce a bucketing scheme, where edges
are placed in buckets based on our estimate of their distance label.
We will then show that our estimates are not too far off, that is, we
will only have to search a small number of buckets to find a full
or deficient edge.  We will also have to introduce a larger gap 
between full and deficient, which will also alter the analysis of
the length of an alternating path.

We now describe the details of our approach. We assume familiarity with the 
Section~\ref{sec:ap-h-small}  and only emphasize the differences. Also, rather than separately dealing with edge reorientations,
we just process the flip of an edge $(u,v)$ as a deletion of the edge, and then an insertion of it oriented
in the opposite direction. Since by by Theorem \ref{thm:dynamic-orientation-general} each change in $G$ causes at
most $O(1)$ edge flips, this only increases the running time by a constant factor.

Recall the parameter $\la$ from Theorem \ref{thm:h-general} and assume for simplicity 
that $\la \be$ is an integral multiple of $6$.
We begin by redefining full and deficient.
\begin{itemize}
\item A {\em full}  edge $(u,v)\in H$ has $\dh(u) + \dh(v) = \beta$,
\item A {\em deficient edge} $(u,v) \in G-H$ has $\dh(u) + \dh(v) =   \be (1 - \la)$.
\end{itemize}

We have, in particular, redefined deficient to be not $\be -1$ but rather
a constant fraction of $\be$.  We add this extra space because we will
no longer be able to maintain degrees exactly, and thus when we
augment, we will no longer be able to alternate between full and
deficient edges but rather between full and a relaxed notion of
deficient.

In order to define this relaxed notion, we introduce the notion of
edge ranges, which capture the various intermediary levels of deficiency and fullness that an edge can have. 
We will think of our edge degrees as being partitioned into 8 ranges $F_i$ ($F$ for different degrees of fullness) defined in terms of a parameter $\ell = \be \la / 6$. We will then refer to an edge as being in one of the ranges, depending on its edge degree (Note that when we say an edge is in one of these ranges, this always refers to the actual edge degree, not to any incorrect estimates we hay have).
\begin{itemize}
\item Range $F_0$ contains edges with edges degree $< \be (1 - \la)$ (such edges cannot be unused, as they would violate Property P2 of an \edcsab.
\item Range $F_7$ contains edges with edge degree $\be$. (These are the full edges.)
\item Range $F_i$, for $1\leq i\leq 6$ contains edges with edge degree in 
[$\beta(1 - \la) + \ell(i-1), \beta(1 - \la) + \ell i$]. 
\end{itemize}

We call {\em unused} edges in $F_1,\ldots,F_5$ {\em augmentable}.  We specifically
omit $F_6$ from the definition of augmentable, even though such edges 
can be used by property P1 of an \edcsab, 
in order to leave a gap between augmentable edges and full edges.

We also recall that edge vertex accurately knows its own degree.  The
inaccuracy comes from the inability of a vertex to inform all its
neighbors, or even all its owned neighbors of its correct degree. 

We now describe the insertion and deletion procedures.  As in Section~\ref{sec:ap-h-small},
we will first describe them at a high level, ignoring the implementation details and
then fill those in later.

\paragraph{Insertion.}   
Consider an edge $(u,v)$ that has just been added to $G$. We know
$\dh(u)$ and $\dh(v)$ exactly and can therefore compute edge degree $\ed(u,v)$   If
$\ed(v,w) \geq \beta (1 - \la) $, then we do not add the edge to $H$, and the 
properties P1 and P2 from Definition~\ref{dfn:unweighted-edcs}
remain satisfied.  (Recall that $\bem = \be (1 - \la)$).  On the other hand, if 
$\ed(v,w) < \be (1 - \la) $, we want to  add $(u,v)$ to $H$.  Doing so 
will increase $d(u)$ and $d(v)$ by 1, which may lead to a violation of
$P1$ for other edges in $H$ that are incident to either $u$ or $v$.  
Thus, we will need to find a type of alternating path that will allow us to
add $(u,v)$ and still maintain P1 and P2 for all vertices.

We now say that a vertex $x$ is {\em increase-safe} if it has no incident
full edges, and say that it is {\em decrease-safe} if it has no
incident augmentable edges. 

 Returning to adding $(u,v)$ to $H$, let's
focus on vertex $v$; we will then deal with vertex $u$ analogously.  If $v$ is increase safe, then when we add
$(u,v)$ we have not violated P1 for any edges incident to $v$, since
there are no full edges incident to $v$.  If $v$ is not increase safe,
then it must have at least one  incident full edge, say $(v, p_1)$.
We would like to add $(u,v)$ to $H$ and remove $(v,p_1)$ to $H$,
thereby leaving $v$'s degree unchanged.  Doing so would decrease
$d_H(p_1)$, which we can do only if $p_1$ is decrease safe.  If $p_1$
is decrease safe, then adding $(u,v)$ to $H$ and removing $(v,p_1)$
leaves $v$'s degree unchanged and decreases $p_1$'s degree and
reestablishes $P1$ and $P2$ for all vertices (except possibly $u$).
However, if $p_1$ is not decrease-safe, it must have an incident
deficient edge, say $(p_1,p_2)$.  We can add this edge to $P$
and continue from $p_2$ we did from $v$.   We can continue in this
manner, stopping when we find either an increase-safe vertex or decrease-safe vertex.
Assume that the set of edges we find form a simple path.  Then we can exchange the role of the matched
and unmatched edges on $P$
thereby reestablishing $P1$ and $P2$ for all
vertices (except possibly $u$).  Note that this path may leave the
number of edges in $H$ unchanged, or may increase the number of edges
in $H$ by one.  Either outcome is acceptable.

We now argue that the set of edges we find do form a simple path.
In addition, in order to bound the time, we would
like to bound the length of $P$.  
We do so with the following lemma:

\begin{lemma}
\label{lem:path-length-general}
Let $P$ be a path of alternating full and augmentable edges.  Then $P$
is simple and the length of $P$ is at most $12/\la +1$.
\end{lemma}

\begin{proof}
Consider first the case that $P = (p_0,p_1,\ldots,..,p_k)$ starts with
a full edge.  Let $d = d_H(p_0)$ and clearly $d\leq \beta$.  Since
$(p_0,p_1)$ is full, $d_H(p_1) = \beta - d$.  Since $(p_1,p_2)$ is
augmentable, we have that $\dh(p_1) + \dh(p_2) \le \beta - \la \be
/6$ which  implies that 
$d_H(p_2) \le  d - \la \be/6$.
Continuing, we get that $d_H(p_3) \geq  \beta -d + \la \be/6 $,
$d_H(p_4) \le  d- 2 \la \be / 6$,
and in general $d_H(p_i) \le  d- i \la \be /3$ for 
even $i$ (Figure~\ref{fig:augpath} contains a similar argument, except
there we just had $\be - 1$ instead of $\beta(1 - \la/6)$.)  Since  $d \leq \be$, we have that after
$\be / (\be \la /6) = 6/ \la$ even vertices, the degree will be
$0$ and the path will terminate.  The length of the path will
therefore be bounded by twice the number of even indexed vertices plus
one for $12 /\la + 1$.
Furthermore, if we consider all the
vertices on $P$ that are on the same side of the bipartite graph, they
all have distinct $\dh$ values, and hence they must be distinct and the
path is therefore simple.

Now consider that $P= (p_0, p_1, \ldots, p_k)$ starts with an
augmentable edge.   
Going through the same argument, we now have that the degrees
of the odd indexed vertices are decreasing, and we have that the
length of the path is at most $12/\lambda + 1$.
\end{proof}


After finding an alternating path from $v$, we repeat the same
procedure starting at $u$.  (Note that it is fine for the path from
$v$ to intersect the path from $u$, as we execute the fixing up
procedures sequentially). We have thus shown the following:

\begin{lemma}
\label{lem:insertion-length-general}
After inserting an edge into G, we can reestablish P1 and P2 using at
most $24 /\la + 2 $ insertions/deletions
from $H$.
\end{lemma}

\paragraph{Deletions}

Deleting an edge $u$ from $G$ is handled in a similar manner to
insertions.  If $(u,v)$ is not in $H$, then we do not need to change
$H$.  If $(u,v)$ is in $H$ and both $u$ and $v$ are
decrease-safe, we just remove the edge $(u,v)$.  Otherwise, we find an
alternating path in the same way we did for insertions and observe
that Lemma~\ref{lem:path-length-general} applies for paths beginning with both
full and augmentable edges.  Thus we have
 
\begin{lemma}
\label{lem:deletion-length-general}
After deleting an edge from G, we can reestablish P1 and P2 using at
most $24/\la + 2$ insertions/deletions
from $H$.
\end{lemma}

\paragraph{Finding Alternating Paths}
Now, in order to implement the augmenting procedure, we need to be
able to accurately identify when a vertex has an incident full  edge
and when it has an incident augmentable edge.   Identifying the incident
full edge is straightforward: a full edge is in $H$, and $\dh(v) \leq
\beta$, so in $O(\beta)$ time one can scan all the incident edges
in $H$, both owned and unowned.  
Identifying an incident augmentable edge is more challenging, and we
will need to introduce several additional ideas and data structures.

\paragraph{Buckets}
Conceptually, we want each vertex to maintain for each of its incident edges
which of the ranges $F_0,\ldots,F_7$ that edge belongs to,
but we will need to do that in an indirect way. The main challenge
that arises is that a vertex may maintain inaccurate information about its
neighbors. Another challenge is that if a vertex $v$ tried to bucket its edge
degrees, or even its estimates for edge degrees, then an increase in $\dh(v)$ would 
cause all of its incident edges to increase in edge degree, and so might require moving
many edges to different buckets.

Our solution is to let each vertex $u$ maintain (possibly inaccurate) information about
the degree of all of its {\em unowned} neighbors by bucketing each neighbor $v$ according to u's estimate of v's
degree, denoted $\tdu(v)$.  We will use buckets of
width $\ell = \beta \la/ 6$, so $u$ has $\beta/\ell = 6 / \la$ buckets $\bu_1, \bu_2, ..., \bu_{6/\la}$.
That is, $\bu_i$ contains neighbors $v$ for which $\ell (i-1) \leq \tdu(v) < \ell i$. We say that
vertex $v$ {\em properly belongs} in bucket $\bu_i$ if $\ell (i-1) \leq \dh(v) < \ell i$, that is
if $\bu_i$ would be v's bucket if $u$ had accurate information about the degree of $v$.

\paragraph{Edge Updates}

Let $r = 18 \sqrt{m}/(\la \beta)$. 
Each vertex $u$ will maintain its {\em owned} edges in a doubly-linked circular list $L^u$  with
two pointers $p$ and $q$.   An {\em information update} consists of informing the
next $r$ edges $(u,v)$ on the list of the accurate value of $\dh(u)$, so that
they can update their bucket structures $B^v$.   The information
update will always start at the pointer $p^v$,  and $p^v$ will advance as
the information update proceeds.  New edges will be
added to the list just before $p^v$ (i.e. to the ``back" of the list).  A second pointer $q^v$ will be the {\em
  repair} pointer, and will point to the next edge to consider including in $H$ (more on this later).

\paragraph{The algorithm} 

We will now describe the algorithm to search for a alternating path of full and augmentable edges.  
\begin{itemize}
\item As observed before, finding a full edge can be done in $O(\beta)$ time.  
\item To find an augmentable edge incident to $u$, we first want to see if 
$u$ has any unowned edges in the range $F_1, ..., F_5$. To do this, we start checking
the buckets of $u$. We start with the bucket whose degree interval contains 
$\beta(1 - \la) - \ell - \dh(u)$,  i.e. the bucket $\bu_i$ such that $\beta(1 - \la) - \ell - \dh(u) \in [\ell(i-1), \ell i)$; 
the reasoning behind these boundaries will becomes evident later.
We then check this bucket $\bu_i$ and then buckets $\bu_{i+1}$ then $\bu_{i+2}$ all the way up to $\bu_{i+7}$
until we find some non-empty bucket. If we don't find a non-empty bucket, we declare that we have failed 
to find an augmentable edge, and move on the next step. 
Otherwise, we pick an arbitrary edge $(u,v)$ from the first non-empty bucket that we find 
(remember, we always from smaller to larger degree buckets)
and check if $(u,v)$ is augmentable by checking if $\dh(u) + \dh(v)$ is in $F_5 or below$; if it is augmentable we 
augment down it, otherwise we do NOT check for more edges but simply declare that we have failed
to find an augmenting edge, and move on to the next step. Note that this whole operation thus takes only $O(1)$ time as
we only check 8 buckets. Note also, however, that it is quite possible for $u$ to not find an unowned augmenting edge
even though it actually has one: the buckets have inaccurate degrees, so $u$ may happen to pick a vertex $v$ from the bucket
that has large degree (and so $(u,v)$ is not augmentable), even though there are other vertices in the bucket with small degree.
\item  If we didn't find an augmentable unowned edge in the previous step, we look at the next $r$ {\em owned} edges of $u$ by starting at the repair pointer $q^v$ and moving forward $r$ steps on the
list.  For each one of these edges, we can exactly compute $\dh(u) + \dh(v)$ and 
check if the edge is augmentable. 
If we find an augmentable edge we stop and take that edge. 
Otherwise, we move on the the next step.
\item  If we make it to this step then we were unable to find an augmentable edge 
(though one may in fact exist), so we allow this vertex to be the end of the augmenting path, 
we increase/decrease its degree accordingly, and we perform an information update on this
vertex (see above).
\end{itemize}

\noindent In order to prove the correctness of the algorithm, we fill prove that
it maintains the following invariants:

\begin{enumerate}
\item Every time vertex $v$'s degree changes, we execute an
  information update at $v$.
	
\item \label{inv:bucket-info} Say that edge $(u,v)$ is owned by $v$ and recall that $\tdu(v)$ 
is u's estimate of the degree of $v$ and $\dh(v)$ is the actual degree of $v$. 
Then, we always have $\dh(v) - \ell \leq \tdu(v) \leq \dh(v) + \ell$. 
In particular, since $\ell$ is the range of a bucket $\bu_i$, if $v$ {\em properly} belongs in bucket $\bu_i$ 
then it is in fact contained in one of buckets $\bu_{i-1}, \bu_i$, or $\bu_{i+1}$.

\item When a vertex $u$ checks its unowned neighbors for an augmentable edge, every augmentable edge is in one
of the 8 buckets that $u$ is allowed to look at (though $u$ may not end up checking that particular edge). 

\item \label{inv:find-good-edge} If $u$ has an unowned edge to vertex $v$ of degree $d(v)$, 
then when $u$ looks in its bucket structure
for an unowned augmentable edge, if it picks some edge $(u,w)$ then we must have: $d(w) \leq d(v) + 3\ell$.

\item \label{inv:unowner-cannot-decrease} 
As long as $v$ owns $(u,v)$  and $(u,v)$ is in range $F_2$ or
lower, $d(u)$ cannot decrease.
	
\item \label{inv:owner-cannot-decrease} 
If $v$ owns $(v,w)$ and $d(v)$, over some time range, has
decreased by $\delta$, then at some point in that time range the
edge $(v,w)$ was not augmentable.
	
\end{enumerate}

\paragraph{Proof that the invariants hold}
\begin{enumerate}
\item This clearly holds by the design of our algorithm. 

\item $v$ owns $(u,v)$, so consider the last time $v$ sent its accurate information to $u$ 
during an information update of $v$.
Now, every time the degree of $v$ changes it updates the information of $r$ owned neighbors,
so since by our orientation (Theorem \ref{thm:dynamic-orientation-general}), $v$ owns at most $3 \sqrt{m}$ edges,
the degree of $v$ can change at most $3 \sqrt{m} / r = \beta \la / 6 = \ell$ times before it updates $u$ again, 
so $\tdu(v)$ is off by an additive factor of at most $\ell$, as desired.

\item This follows from the fact that the 8 buckets $u$ is allowed to
  look at span all degrees between $\beta(1 - \lambda) - \ell -
  \dh(u)$ and $\beta + \ell - \dh(u)$, so since by Invariant
  \ref{inv:bucket-info}, u's degree information about a vertex is at
  most one bucket off, any vertex with actual degree $\dh(v)$ between
  $\beta(1 - \lambda) - \dh(u)$ and $\beta - \dh(u)$, so all edges
  with edge degree between $\beta(1 - \la)$ and $\beta$ -- which
  includes all augmentable edges -- are in one of these 8 buckets.

\item By Invariant \ref{inv:bucket-info} if $\dh(v)$ properly belongs
  in bucket $\bu_i$ then in reality $\tdu(v)$ will be in bucket at
  most $\bu_{i+1}$. Thus, since we always check lower indexed buckets
  first, the chosen edge $(v,w)$ will be to a vertex $w$ for which
  $\tdu(w)$ is in bucket at most $\bu_{i+1}$. Applying Invariant
  \ref{inv:bucket-info} again, $\dh(w)$ is in bucket at most
  $\bu_{i+2}$, so since $\dh(v)$ was in bucket $\bu_i$ and the size of
  each bucket is $\ell$, the difference between $\dh(v)$ and $\dh(w)$
  is at most 3$\ell$.

\item The degree of $u$ can only decrease if when $u$ searches for an
  unowned augmentable edge, and it does not find one. By invariant
  \ref{inv:find-good-edge}, if during its search for an augmentable
  edge, $u$ picks $(u,w)$ then $\dh(w) \leq \dh(v) + 3\ell$. But since
  by assumption $(u,v)$  is in $F_2$ or lower, we must have that $(u,w)$ is
  in $F_5$ or lower (the size of each range $F_i$ is precisely $\ell$), so
  $(u,w)$ is augmentable and the degree of $u$ does not decrease.

\item Every time the degree of $v$ decreases it scans $r$ edges in its
  repair list. Thus, by the time its degree changes by $\ell$ it has
  scanned $r \ell = 3\sqrt{m}$ edges, so since by our orientation
  algorithm every vertex owns at most $3\sqrt{m}$ edges (Theorem
  \ref{thm:dynamic-orientation-general}), and $v$ must have reached
  $(v,w)$; at that point, either $(v,w)$ was already not augmentable
  and we are done, or it was augmentable in which case $v$ would
  augment down it and $(v,w)$ would become used and hence not
  augmentable.
\end{enumerate}

\noindent Using the above invariants we the algorithm always maintains
the properties of an \edcs. First, since the algorithm explicitly
checks for full edges we never increase the degree of a vertex that
has an adjacent full edge, and so property P1 of an \edcsab\ is always
maintained (see Definition \ref{dfn:unweighted-edcs}).  To verify
property P2, let us say for contradiction that there is some unused
edge $(u,v)$ such that $\dh(u) + \dh(v) < \be(1 - \la)$, i.e. such
that $(u,v)$ is in $F_0$. Let us say, wlog, that $v$ owns $(u,v)$.
Note note that at some point in the sequence $(u,v)$ must have been
not augmentable: if it was augmentable when it was inserted, then the
algorithm would have augmented down it, and the only way it could
become unused is if it was augmented down when full, in which case it
would be in $F_6$ and hence not augmentable. There must exist a last
time that $(u,v)$ dropped from $F_3$ to $F_2$ , i.e. the time after
which it was always in $F_2$ or below. By Invariant
\ref{inv:unowner-cannot-decrease} the degree of $u$ cannot have
decrease since that point in time. By Invariant
\ref{inv:owner-cannot-decrease}, the degree of $v$ can decrease by at
most $\ell$. Thus, edge $(u,v)$ cannot be below range $F_1$, which
contradicts our assumption that it was in $F_0$.

We conclude with a summary and proof of Theorem~\ref{thm:h-general}.

\begin{proofof}{Theorem~\ref{thm:h-general}}
From Lemma~\ref{lem:path-length-general}, we have that the update ratio is
$O(1/\lambda)$.  To bound the running time, we examine the time 
to find a path of length $O(1/\lambda)$. 
each search takes either $O(\beta)$
time for a full edge, or to find an augmentable edge
$O(\sqrt{m}/(\la \be))$ time.  This gives a total time of
$O( \left( \frac{1}{\lambda} \right) \left(\beta +
\frac{\sqrt{m}}{\la \be} \right) )$ time.

\end{proofof}

\section{Dynamic Orientation}
\label{sec:ap-sqrt-orientation}

In this section we formally state the dynamic orientation results used by our algorithm, and prove Theorem \ref{thm:dynamic-orientation-general}, which
is new to this paper. 

\begin{theorem} [\cite{KopelowitzKPS14}]
\label{thm:dynamic-orientation-arb}
Let $G$ be a graph that always has arboricity at most $\alpha$. We maintain an orientation, under edge insertions and deletions, with the following properties: the maximum load at all times is $O(\alpha + \log n)$, the worst-case number of flips per insertion/deletion is also $O((\al + \log(n)))$, 
and the worst-case time to process an insertion/deletion in $G$ is $O(\al(\al + \log(n)))$.
(If we do not have an upper bound on $\al(G)$ in advance there is a variant whose bounds are in terms of the exact
arboricity of the current graph, at the expense of an extra $\log(n)$ factor).
\end{theorem}

\begin{theorem}
\label{thm:dynamic-orientation-general}
 In a graph $G$, we can maintain an orientation, under insertions and deletions, with the following
properties: the max load at all times is at most $3\sqrt{m}$, the worst-case number of flips per insert/deletion in $G$
is $O(1)$, and the worst-case time spent per insertion/deletion in $G$ is $O(1)$.
\end{theorem}

\begin{proof}
For simplicity of analysis, we assume that we begin with a graph with no edges, and update from there. Let us start with a few definitions.

\begin{definition} 
Define a vertex to be {\em small} if it has degree (not load) less
than $2\sqrt{m}$ and {\em large} if it has degree greater than or
equal to $2\sqrt{m}$. Given some orientation, define a vertex in the
current orientation to be {\em heavy} if it has load greater than $2
\sqrt{m}$.
\end{definition}

\begin{observation}
\label{obs:sqrt-degree}
A graph can contain at most $\sqrt{m}$ large vertices, and at most $2\sqrt{m}$ vertices of degree $\geq \sqrt{m}$. Otherwise, the total degree of these vertices would be greater than $2\sqrt{m}\sqrt{m} = 2m$, so the number of edges in the graph would be greater than $m$.
\end{observation}

\noindent The above observation makes it clear why given any graph we can compute a  $2\sqrt{m}$-orientation in linear time. Simply let small vertices own all of their edges; if an edge is between two small vertices or two large vertices, it can go either way. Now no vertex can have load greater than $2\sqrt{m}$, as then all its owned edge would go to large vertices, which would contradict the observation above.

\begin{observation}
\label{obs:edge-total}
The natural way for a vertex to transition from small to large, or from heavy to non-heavy, is due to the insertion or deletion of edges incident to this vertex. But because all of these terms are defined in terms of $\sqrt{m}$, a vertex can also transition simply because the number of edges has changed, and so $\sqrt{m}$ has changed. This is certainly not a big deal as $\sqrt{m}$ changes very slowly, but it is inconvenient for our analysis as we would like to treat $\sqrt{m}$ as a fixed number. We handle this using a standard technique in dynamic algorithms: for $\sqrt{m}$ to double, $m$ would have to increase by a factor of 4, and $4m = O(m)$ time is enough to slowly construct a new orientation from scratch in the background (by staggering the linear-time algorithm above over many updates). For this reason, we can assume that the number of edges is fixed within a factor of $2$, and let $m$ refer to the upper bound of this range, thus allowing us to treat $\sqrt{m}$ as fixed.
\end{observation}

\noindent Observation \ref{obs:sqrt-degree} provides a very simple algorithm for maintaining a $3\sqrt{m}$-orientation in {\em amortized} update time $O(1)$. The algorithm is as follows: when a vertex reaches load above $3\sqrt{m}$, simply scan all of its edges and flip edges going to small vertices. By Observation \ref{obs:sqrt-degree}, fewer than $2\sqrt{m}$ of its edges went to large vertices, so we have transformed the vertex into a non-heavy vertex, without adding any heavy vertices. This transformation requires $3\sqrt{m}$ time to scan all neighbors, but it takes $\sqrt{m}$ to turn a non-heavy vertex back into one with degree greater than $3\sqrt{m}$, so we only need $3$ credits per update.

Worst-case update time is only slightly more difficult. Whenever an
edge is inserted, if exactly one of the endpoints is small we give
that endpoint ownership: otherwise, we assign ownership
arbitrarily. Every time the load of a heavy vertex increases, we will
scan 5 edges that it {\em owns}, and flip any edges going to small
vertices.  From the other direction, every time the degree of a small
vertex decreases we will scan 5 of its edges (owned and not owned) and
automatically flip them towards the small vertex. Scanning will occur
in a round-robin fashion: each vertex just stores a list of outgoing
edges and another list of outgoing {\em owned} edges, and moves along this
list; if a new edge is inserted into one of the lists, we put it at
the ``back of the list'', i.e. right behind the current pointer.  We will now show that this algorithm maintains a
$3\sqrt{m}$-orientation.

\begin{invariant}
\label{inv:small-degree}
As $G$ changes, at no point in time can a vertex $u$ of load more than $3\sqrt{m}$ own an edge $(u,v)$ where $v$ has degree (not load) less than $\sqrt{m}$
\end{invariant}

\begin{proof}
For the sake of contradiction, let us consider the first time that $u$
owns such an edge $(u,v)$. Let $t$ be the last time this edge flipped
-- i.e. the time after which $u$ always owned the edge. Now, it is
clear from our algorithm design that when the flip occurred it could
not have been the case that $v$ was small and $u$ was heavy and yet
ownership was given to $u$. So at time $t$ either $v$ was large or $u$
was not heavy, but by the current (to-be-contradicted) time, $v$ is
small and $u$ is heavy. Let $t^*$ be the very last time that $v$
transitioned from large to small, or $u$ transitioned from non-heavy
to heavy, whichever came later. If $v$ transitioning from large to
small came later, then $v$ must have experienced at least $\sqrt{m}$
degree decreases since time $t^*$ (it dropped from degree $2\sqrt{m}$
to degree $\sqrt{m}$), and over this time it scanned at least
$5\sqrt{m}$ edges. Since scanning is done round-robin, it would
certainly have examined edge $(u,v)$ during one of its scans, and
would have flipped it because small vertices always flip --
contradiction. But similarly, if the later event was $u$ becoming
heavy, then $u$ must have experienced at least $\sqrt{m}$ degree
increases (from $2\sqrt{m}$ to $3\sqrt{m}$), over which time it
scanned 5$\sqrt{m}$ edges and so would certainly have scanned $(u,v)$
and would have flipped it because after time $t^*$ $v$ was small --
again a contradiction.
\end{proof}

The above invariant shows that as our algorithm processes updates to $G$, no vertex can ever have load above $3\sqrt{m}$, as then all of its neighbors would have degree more than $\sqrt{m}$, which is impossible by Observation \ref{obs:sqrt-degree}.

\end{proof}

\newpage

\ignore{
\begin{figure}
\includegraphics[scale=0.5]{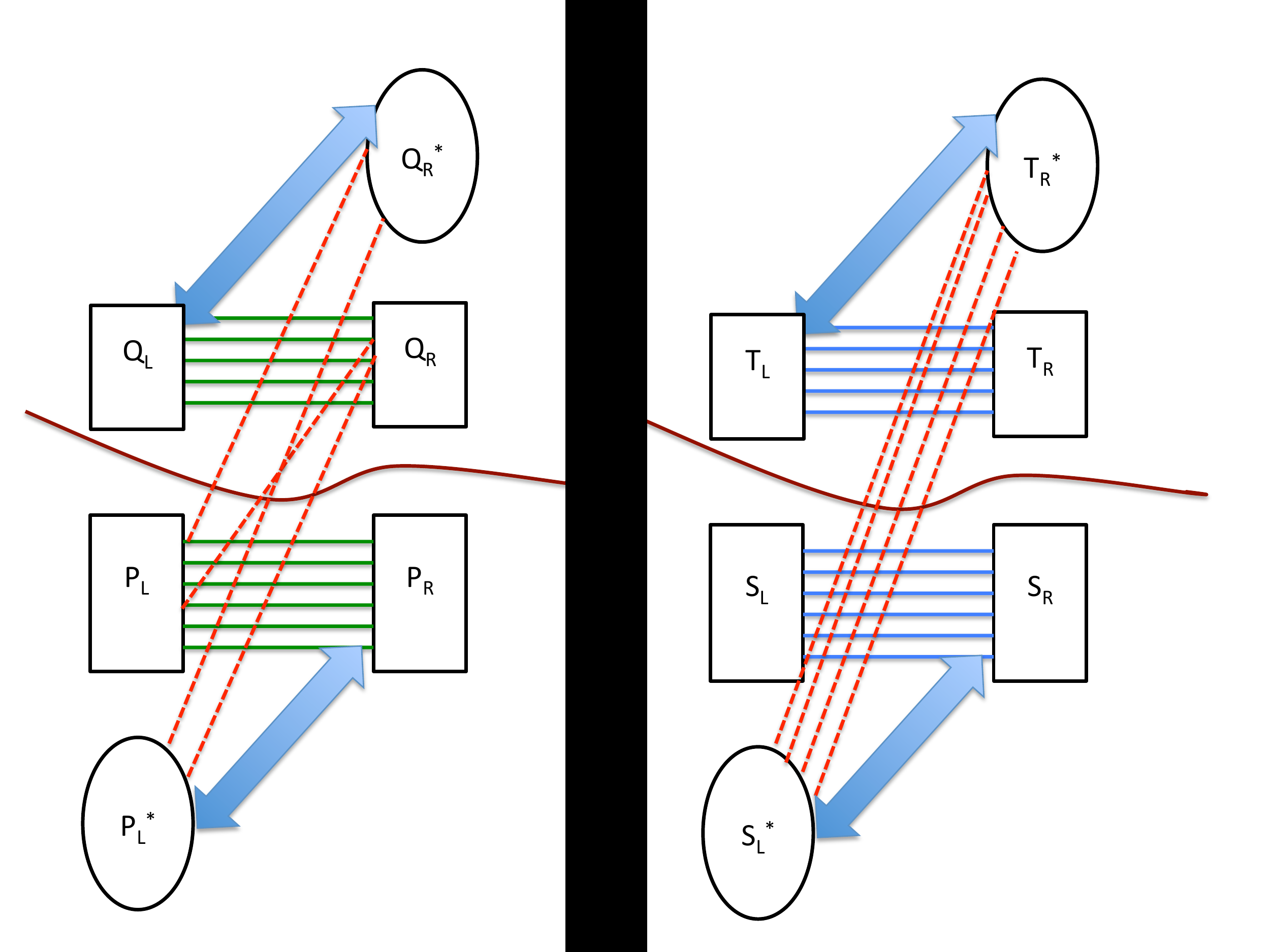}
\caption{The right side of the figure is the desired $(S,T)$ split of Lemma \ref{lem:cut-edges}. 
The left side of the figure is the standard $(P,Q)$ cut induced by a matching. The figure
illustrates both their similarities and their differences; the key argument
of Lemma \ref{lem:cut-edges} is that we can fixup the left side to look like the right. On the left side, the starred sets
are free vertices, the P side of the cut is everyone reachable from $\pls$ in the residual graph
, and the Q side of the cut is everyone else.
The red dashed edges are edges in $G$ but not $H$; there must be such edges crossing the cut,
but they need not be between $\pls$ and $\qrs$: they can go anywhere from one side of the cut
to the other. The single green edges represent a matching in $H$ between $\ql$ and $\qr$ 
and between $\pl$ and $\pr$. The thick blue lines represent that there could
be a bunch of edges, in $G$ or in $H$, between the corresponding sets.
The right side is different in two crucial ways. The first is that all the red edges in $G - H$ 
must go directly between the two starred sets. The second is that the matching edges
between $\tl$ and $\tr$ and between $\sl$ and $\sr$ are now blue to represent that
there may no longer be a matching between them in $H$ (green), but there is still guaranteed to be 
a matching between them in $G$ and possibly $H$ (blue).}
\label{fig:pqcut}
\end{figure}
} 

\begin{figure}
\includegraphics[scale=0.5]{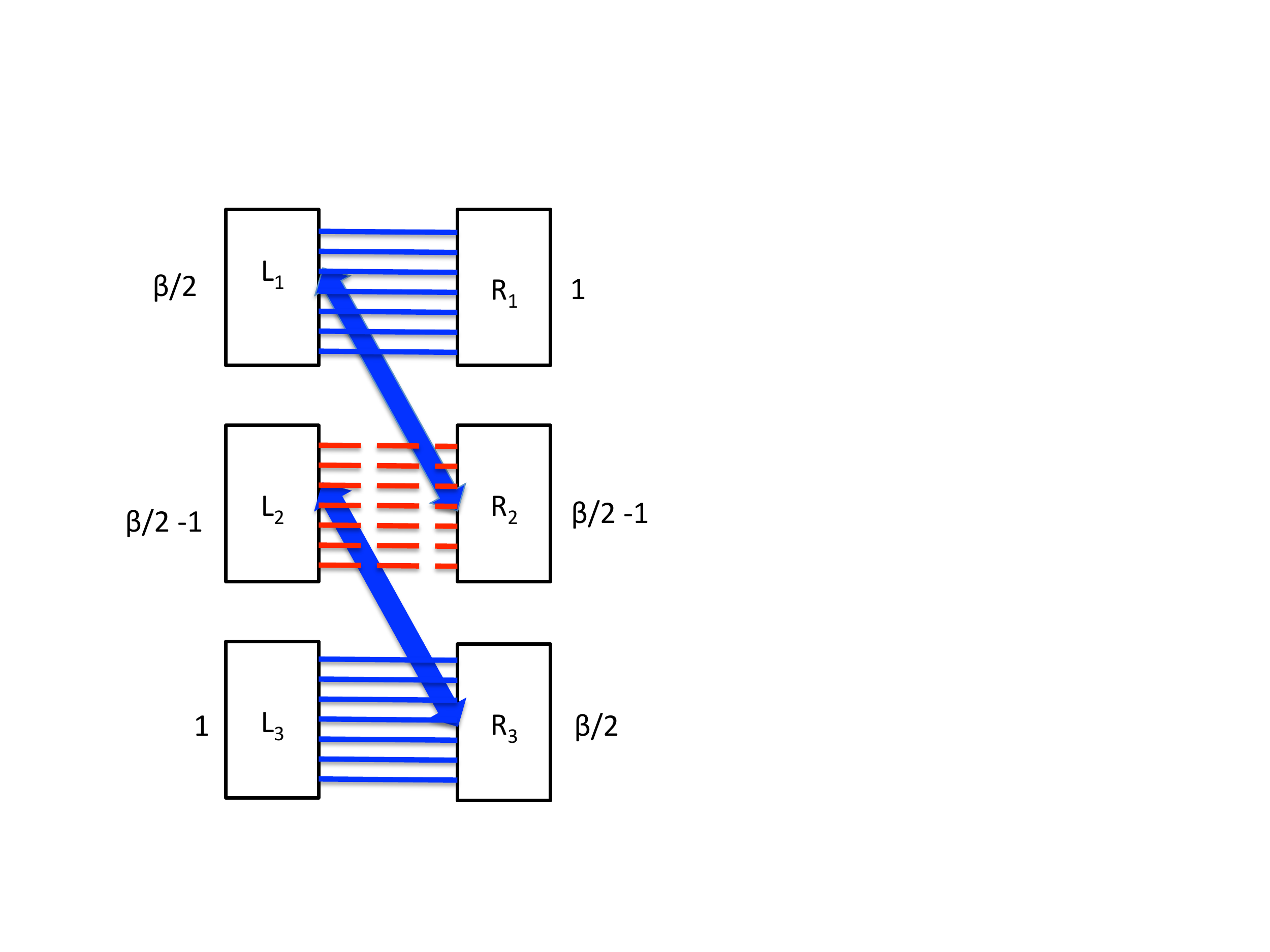}
\caption{In this example, we see the problem that arises with an {\em unweighted} EDCS.  
Each side of the bipartite graph is split into 3 equal sized
pieces.  The thick blue arrows represent bipartite graphs of degree $\be/2 - 1$, 
while the other blue edges signify a matching.  The blue edges are in H
whereas the red dashed edges are in G but not H.  A maximum matching in H matches only 2/3 of the vertices, whereas 
a maximum matching in $G$ matches all of them.  The values of $\dh$ are written next to the 
edge blocks.  We see that it is legal to omit the dashed red edges from $H$, 
since their total degree is more than $\be (1 - \la)$.  
If we had a weighted EDCS, however,
we would be forced to increase the degrees of $R_1$ and $L_3$ by adding 
multiple copies of edges between $L_1$ and $R_1$ and between $L_3$ and $R_3$.  
Thus the degrees of $L_1$ and $R_3$ would increase, which would force 
the degrees of $L_2$ and $R_2$ to decrease (property P1 of a weighted \edcsab),
but then the red edges would defy property P2 of a weighted \edcsab\ so we would
have to add some of them to $H$, and hence increase the size of the maximum matching in $H$.}
\label{fig:bad}
\end{figure}

\begin{figure}
\includegraphics[scale=0.5]{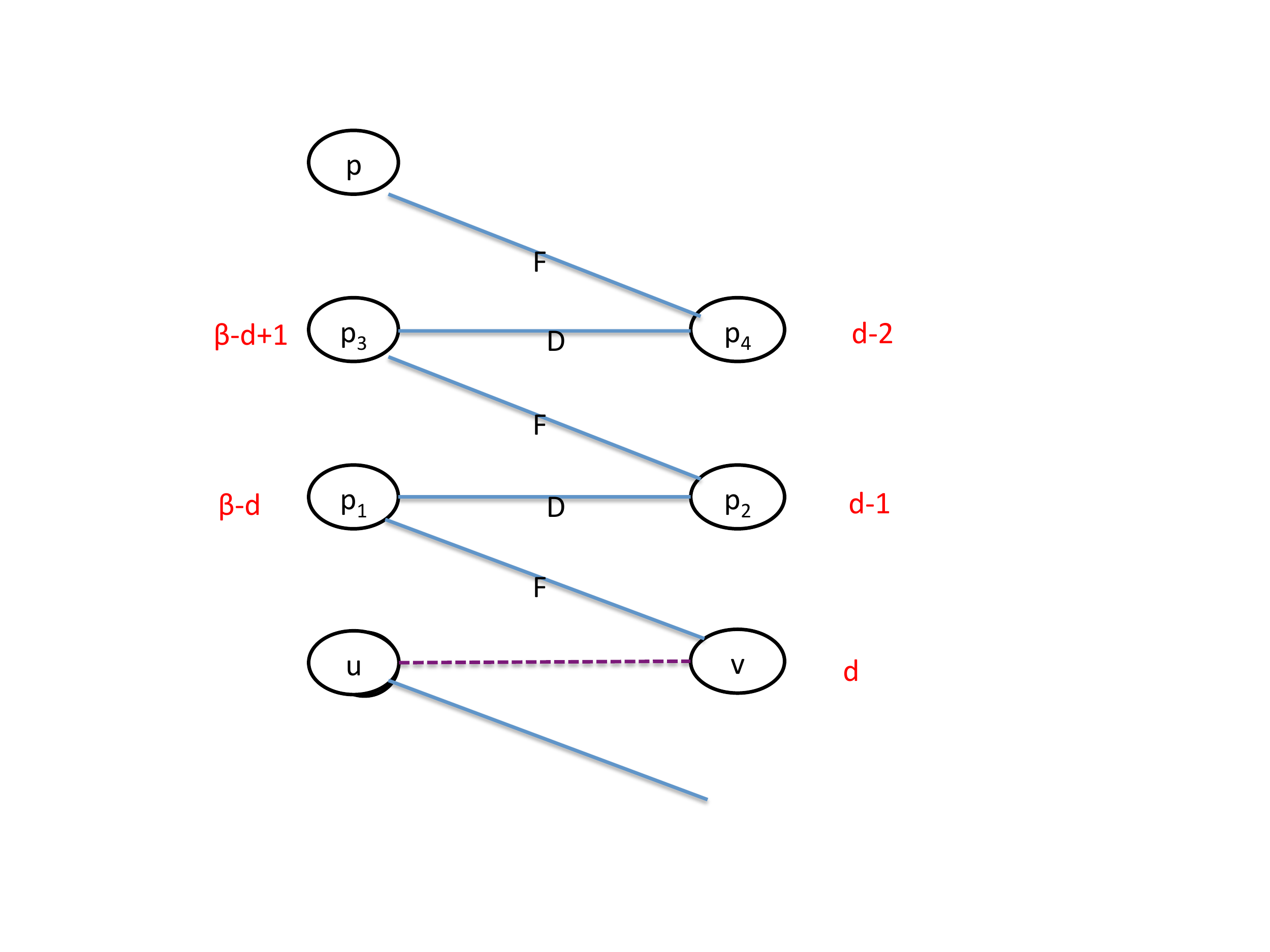}
\caption{An illustration of the alternating path.  We see the edge degree in red next to the vertices.  Across deficient (D) edges,
they sum to $\beta -1$ and across full (F) edges, they sum to $\beta$.  We see that the distances on the right side are decreasing 
along the path, while those on the left are increasing.}
\label{fig:augpath}
\end{figure}

\end{document}